\documentclass{article}

\usepackage{amsfonts,amsmath,amssymb,amsthm}
\usepackage{epic,eepic}

\usepackage{graphicx}

\usepackage{fullpage}

\newcommand{\vvec}{\mathbf{v}}
\newcommand{\Avec}{\mathbf{A}}
\newcommand{\Bvec}{\mathbf{B}}
\newcommand{\Gvec}{\mathbf{G}}
\newcommand{\Hvec}{\mathbf{H}}

\newcommand{\Ivec}{\mathbf{I}}
\newcommand{\Jvec}{\mathbf{J}}
\newcommand{\Kvec}{\mathbf{K}}

\newcommand{\uvec}{\mathbf{u}}
\newcommand{\xvec}{\mathbf{x}}

\newcommand{\Zerovec}{\mathbf{0}}
\newcommand{\Onevec}{\mathbf{1}}

\newcommand{\Lamvec}{\mathbf{\Lambda}}

\newcommand{\sech}{\mathrm{sech} \,}

\newcommand {\R}{\mathbb{R}}

\DeclareMathOperator{\Fix}{Fix}

\newtheorem{Theorem}{Theorem}
\newtheorem{Lemma}{Lemma}

\usepackage{color}

\begin{document}

\title{Symmetries constrain dynamics in a family of balanced neural networks}
\author{A.K. Barreiro\footnotemark[1] \footnotemark[4] \and J.N. Kutz\footnotemark[2] \and E. Shlizerman\footnotemark[2] \footnotemark[3] \footnotemark[5]}


\maketitle

\renewcommand{\thefootnote}{\fnsymbol{footnote}}
\footnotetext[1]{Department of Mathematics, Southern Methodist University}
\footnotetext[2]{Department of Applied Mathematics, University of Washington} 
\footnotetext[3]{Department of Electrical Engineering, University of Washington} 
\footnotetext[4]{Supported by a Mathematical Biosciences Institute Early Career Award}
\footnotetext[5]{Supported by NSF/NIGMS DMS-1361145 and Washington Research Foundation Fund for Innovation in Data-Intensive Discovery}
\renewcommand{\thefootnote}{\arabic{footnote}}

\begin{abstract}
We examine a family of random firing-rate neural networks in which we enforce the neurobiological constraint of Dale's Law --- each neuron makes either excitatory or inhibitory connections onto its post-synaptic targets. We find that this constrained system may be described as a perturbation from a system with non-trivial symmetries. We analyze the symmetric system using the tools of equivariant bifurcation theory, and demonstrate that the symmetry-implied structures remain evident in the perturbed system.  In comparison, spectral characteristics of the network coupling matrix are relatively uninformative about the behavior of the constrained system. 
\end{abstract}

\noindent
\textbf{Key words:} recurrent networks, random network, bifurcations, equivariant, symmetry\\

\noindent
\textbf{AMS subject classifications:} 15B52, 34C14, 34C23, 37G40, 92B20

\section{Introduction}
Networked dynamical systems are of growing importance across the physical, engineering, biological and social sciences.  
Indeed, understanding how network connectivity drives network functionality is critical for understanding a broad range of modern-day systems including  the power grid, communications networks, the nervous system and social networking sites.  All of these systems are characterized by a large and complex graph connecting many individual units, or nodes, each with its own input--output dynamics.  In addition to the node dynamics, how such a system operates as a whole 
will depend on the structure of its connectivity graph~\cite{Watts:1998db,Park2013science,Hu2016}, but the connectivity is often so complicated that this structure-function problem is difficult to solve.  

Regardless, a ubiquitous observation across the sciences is that meaningful input/output of signals in high-dimensional networks are often encoded in low-dimensional patterns of dynamic activity. This suggests that a central role of the network structure is to produce low-dimensional representations of meaningful activity.  Furthermore, since connectivity also drives the underlying bifurcation structure of the network-scale activity, and because both this activity and the relevant features of the connectivity graph are low-dimensional, such networks may admit a tractable structure-function relationship. Interestingly, the presence of low-dimensional structure may run counter to the intuition provided by the insights of random network theory, which has otherwise proven to be a valuable tool in analyzing large networks. 

In considering an excitatory-inhibitory network inspired by neuroscience, we find a novel family of periodic solutions that restrict dynamics to a low-dimensional attractor within a high-dimensional phase space.  These solutions arise as a consequence of an underlying symmetry in the mean connectivity structure, and can be predicted and analyzed using equivariant bifurcation theory.  We then show that low-dimensional models of the high-dimensional network, which are more tractable for computational bifurcation studies, preserve all the key features of the bifurcation structure. Finally, we demonstrate that these dynamics differ strikingly from the predictions made by random network theory, in a similar setting.


Random network theory --- in which one seeks to draw conclusions about an ensemble of randomly-chosen networks, rather than a specific instance of a network --- is particularly relevant to neural networks because such networks are large, under-specified (most connections cannot be measured), and heterogenous (connections are variable both within, and between, organisms).  It is particularly tempting to apply the tools of random matrix theory to the connectivity graph, as the spectra of certain classes of random matrices display universal behavior as network size $N \rightarrow \infty$ \cite{tao_etal_2010}. 
The seminal work of Sompolinsky et al. \cite{SompCris88} analyzes a family of single-population firing-rate networks in which connections are chosen from a mean zero Gaussian distribution: in the limit of large network size ($N \rightarrow \infty$), they find that the network transitions from quiescence to chaos as a global coupling parameter passes a bifurcation value $g^{\ast} = 1$. This value coincides with the point at which the spectrum of the random connectivity matrix exits the unit circle \cite{girko85,SommCrisSompStein88,bai97}, thereby connecting linear stability theory with the full nonlinear  dynamics. 

Developing similar results for structured, multi-population networks has proven more challenging. One natural constraint to introduce is that of \textit{Dale's Law}: that each neuron makes either excitatory or inhibitory connections onto its post-synaptic targets. For a neural network, this constraint is manifested in a synaptic weight matrix with single-signed columns.  If weights are tuned so that incoming excitatory and inhibitory currents approximately cancel (i.e. $\sum_j \Gvec_{ij} \approx 0$), such a network may be called \textit{balanced} (we note that our use of the word ``balanced" is distinct from the dynamic balance that arises in random networks when excitatory and inhibitory synaptic currents approximately cancel, as studied by \cite{vvSomp1996,renart10} and others). Rajan and Abbott \cite{RA06} studied balanced rank 1 perturbations of Gaussian matrices and found that, remarkably, the spectrum is unchanged. More recent papers have addressed the spectra of more general low-rank perturbations \cite{Wei12,tao2013,muir_MF_2015}, general deterministic perturbations \cite{Ahmadian_etal_2015}, and block-structured matrices \cite{Aljadeff_2015}.  

However, the relationship between linear/spectral and nonlinear dynamics appears to be more complicated than in the unstructured case. Aljadeff et al. \cite{Aljadeff_2015} indeed find that the spectral radius is a good predictor of qualitative dynamics and learning capacity in networks with block-structured variances. Others have studied the large network limit, but when mean connectivity scales like $1/N$ (smaller than the standard deviation $1/\sqrt{N}$): therefore as $N \rightarrow \infty$, the columns cease to be single-signed \cite{hermann_etal_2012,cabana_touboul_2013,kadmon_HS_2015}. In a recent paper which studies a balanced network with mean connectivity $1/\sqrt{N}$, the authors find a slow noise-induced synchronized oscillation that emerges when a special condition (perfect balance) is imposed on the connectivity matrix \cite{delMolino_etal_PRE_2013}.  As a growing body of work has continued to connect qualitative features of nonlinear dynamics and learning capacity \cite{SusAbb09,RAS10,ostojic2014}, it is crucial to continue to further develop our understanding of how complex nonlinear dynamics emerges in structured, heterogenous networks.

In this paper, we study a family of excitatory-inhibitory networks in which both the mean and variability of connection strengths scale like $1/\sqrt{N}$. In a small, but crucial difference from other recent work \cite{RA06,delMolino_etal_PRE_2013}, we reduce self-coupling. We will show that with this change, these networks exhibit a (heretofore unreported) family of periodic solutions.  These solutions arise as a consequence of an underlying symmetry in the mean connectivity structure, and can be predicted and  analyzed using equivariant bifurcation theory. We show through concrete examples that these periodic orbits can persist in heterogeneous networks, even for large perturbations.  Moreover, we demonstrate that low-dimensional models (reduced order models) can be generated to characterize the high-dimensional system and its underlying bifurcation structure; we use the reduced model to study these oscillations as a function of system size $N$. Thus the work suggests both how biophysically relevant symmetries may play a crucial role in the observable dynamics, and also how reduced-order models can be constructed to more easily study the underlying dynamics and bifurcations.




\section{Mathematical Model}
We consider a network in which each node represents the firing rate of a single neuron, connected by sigmoidal activation functions through a random weight matrix. 
This is the model studied in Refs. ~\cite{SompCris88,RA06,delMolino_etal_PRE_2013}, with some important modifications which we detail below.
%
Specifically, we analyze the family of random networks:
\begin{eqnarray}
\dot{\xvec} & = & -\xvec + \Gvec  \tanh \left(g \, \xvec \right)  \label{eqn:Gsys}
\end{eqnarray}
where 
\begin{eqnarray}
\sqrt{N} \Gvec & = & \Hvec + \epsilon \Avec.         \label{eqn:G_def}   
\end{eqnarray}
$\Hvec$ is an $N \times N$ matrix s.t.
\begin{eqnarray}
\Hvec_{ij} & = &  \left\{  \begin{array}{l l} \mu_E, & \qquad j \leq n_E, j \not= i\\
b_E \mu_E, & \qquad  j \leq n_E, j = i\\
\mu_I, & \qquad n_E < j \leq N, j \not= i\\
b_I \mu_I, & \qquad n_E < j \leq N, j = i  \end{array}   \right.
\end{eqnarray}
and 
\begin{eqnarray}
\Avec_{ij} & \sim &  \left\{  \begin{array}{l l} N(0,\sigma_E), & \qquad j \leq n_E, j \not= i\\
N(0,\sigma_I), & \qquad n_E < j \leq N, j \not= i\\
 0, & \qquad j = i  \end{array}   \right.
\end{eqnarray}
We will use the parameter $f$ to identify the fraction of neurons that are excitatory; i.e. $f = n_E/N$. 
The parameter $\alpha$ characterizes the ratio of inhibitory-to-excitatory synaptic strengths: $\mu_I = -\alpha \mu_E$.
We refer to the network as \textit{balanced} (the mean connectivity into any cell is zero) if $\alpha = \frac{f}{1-f}$; it is inhibition-dominated if $\alpha >  \frac{f}{1-f}$.
In all cases below, $f=0.8$ reflecting the approximately 80 \% / 20 \% ratio observed in cortex; the corresponding value of $\alpha$ for a balanced network is $\alpha = 4$. 
Finally, we choose $\sigma_E, \sigma_I$ so that the variance of excitatory and inhibitory connections into each cell is equal; i.e. $\sigma_E^2 f = \sigma_I^2 (1-f)$. 

The matrix $\Hvec$ has constant columns, except for the diagonal, which reflects self-coupling from each cell onto itself. The parameters $b_E$ and $b_I$ give the ratio of self- to non-self connection strengths, for excitatory and inhibitory cells respectively. We will assume that the effect of self-coupling is to reduce connection strengths; that is, $0 \le b_E, b_I \le 1$.

We note that as in \cite{SompCris88} --- but in contrast to later work \cite{RA06,delMolino_etal_PRE_2013} --- 
self-interactions can differ from interactions with other neurons: i.e. $\Gvec_{jj} \not= \Gvec_{ij}$. This is a reasonable assumption, if we conceptualize each firing rate unit $x_j$ as corresponding to an individual neuron; while neurons can have self-synapses (or \textit{autapses} \cite{ConLees_JPhys_2010}), refractory dynamics would tend to suppress self-coupling from influencing the firing rate.

We will find that many features of the resulting dynamics may be connected to an approximate symmetry of the system.
Specifically, if you remove the ``noise" from the connectivity matrix $\Gvec$  --- so that
$\Gvec_{ij} = \mu_E/\sqrt{N}$ if $j \le n_E, i \not= j$, and $\Gvec_{ij} = \mu_I/\sqrt{N}$ if $j > n_E, i \not = j$ --- then the subspace in which all $E$ neurons have the same activity, and all $I$ neurons have the same activity;
\begin{eqnarray*}
x_j & = & x_E, \qquad j \le n_E\\
x_j & = & x_I, \qquad j > n_E
\end{eqnarray*}
is invariant under the dynamics $\dot{\xvec} = -\xvec + \Gvec \tanh(g \xvec)$. To be precise, the system of equations is \textit{equivariant} under the group of permutation symmetries $(S_{n_E} \oplus S_{n_I})$, which contains any permutation of the $n_E$ excitatory neurons and any permutation of the $n_I$ inhibitory neurons.

We will begin by considering the ``noise-less" system in \eqref{eqn:Gsys}, where $\sqrt{N} \Gvec = \Hvec$. The solutions that arise in this system can be readily identified because of the underlying symmetries of the network. We will find that these solutions actually do arise in numerical simulations: furthermore, they persist even when the symmetry is perturbed ($\sqrt{N} \Gvec = \Hvec + \epsilon \Avec$).

\subsection{Some preliminary analysis: spectrum of $\Hvec$}
To analyze stability and detect bifurcations, we will frequently make reference to the Jacobian of \eqref{eqn:Gsys}, \eqref{eqn:G_def}; when $\epsilon = 0$, we will find that this always takes on a column-structured form.   We begin by summarizing some facts about the spectra of these matrices.

Let $\Kvec_N$ be the matrix of all ones except for on the diagonal; i.e.
\begin{eqnarray}
\Kvec_N & = & \mathbf{1_N} \mathbf{1_N}^T - \Ivec_N
\end{eqnarray}

\begin{Lemma}
$\Kvec_N$ has the following eigenvalues: $\lambda_0 = N-1$, and $\lambda_j = -1$ with geometric and algebraic multiplicity $N-1$.\\
\end{Lemma}
\begin{proof}
By inspection, $\mathbf{1}$ is an eigenvector with corresponding eigenvalue $N-1$ (as each row sums to $N-1$). The remaining eigenvectors must satisfy
\begin{eqnarray*}
\sum_{j \not = i} v_j & = & \lambda v_i \rightarrow\\
\sum_{j} v_j & = & \lambda v_i + v_i\\
& = & (\lambda + 1)v_i = 0,
\end{eqnarray*}
since each such eigenvector is orthogonal to $\mathbf{1}$.
\end{proof} 

The Jacobian of \eqref{eqn:Gsys} has the following special structure: except for its diagonal, the entries in column $j$ depend only on the $j$-th coordinate (and are all equal). This leads to a simplification of the spectrum when the cells are divided into synchronized populations. To be precise, we can make the following statement.\\
%
\begin{Lemma}
Assume we can divide our cells $j=1, \cdots, N$ into $K+1$ populations, where $I_k$ identifies the index set of population $k$, for $k = 0, \cdots, K$. Let $\Jvec$ be $-\Ivec + \Kvec_N \Lamvec + \Bvec$, where $\Lamvec$ and $\Bvec$ are diagonal matrices with 
\begin{eqnarray}
\Lamvec_{jj} & = &   a_k,  \qquad j \in I_k\\
\Bvec_{ij} & = & b_k,  \qquad j \in I_k
\end{eqnarray}
That is, $\Jvec$ has constant columns (except for the diagonal), with the value in each column determined by the population identity.
Then the eigenvalues of $\Jvec$ are:
\begin{enumerate}
\item For each $k=0,...,K$: $-1-a_k + b_k $, with multiplicity $n_{I_k}-1$
\item The $K+1$ remaining eigenvalues coincide with the eigenvalues of the matrix $\tilde{\Jvec}$:
\begin{eqnarray}
\tilde{\Jvec}_{ij} & = & \left\{  \begin{array}{l l} n_{I_j} a_j, & \qquad  j \not=i\\
-1+(n_{I_j}-1) a_j + b_j, & \qquad  j =i  \end{array}   \right.
\end{eqnarray}
where $n_{I_j}$ is the number of cells in population $j$. We note that the size of $\tilde{\Jvec}$ is set by the number of subpopulations; that is,  $\tilde{\Jvec} \in \mathbb{R}^{(K+1) \times (K+1)}$.
\end{enumerate}
\end{Lemma}
\begin{proof}
This can be checked by direct computation:
\begin{enumerate}
\item For $k=0,...,K$: there are $n_{I_k}-1$ linearly independent eigenvectors given by vectors that (a) have support only on $I_k$ and that (b) sum to zero: i.e. $\vvec^k_j = 0$ if $j \notin I_k$; and $\vvec^k \perp \Onevec$. 
\item The remaining eigenvectors are given by vectors that are constant and non-zero on each index set: $\vvec_j = c_k$ if $j \in I_k$, and $\left[ \begin{array}{llll} c_0 & c_1 & \cdots & c_K \end{array} \right]$ is an eigenvector of $\tilde{\Jvec}$.
\end{enumerate}
\end{proof}

\noindent 
We consider specific examples that are of particular importance: \\

\noindent
\textbf{Example 1}: 
We consider a balanced network with (possibly) reduced self-coupling: $\alpha = \frac{n_I}{n_E}$ and $0 \le b_E, b_I \le 1$.
The origin $\xvec = \Zerovec$ is a fixed point of \eqref{eqn:Gsys} for all $g$. Therefore, we can think of the population as consisting of two synchronized populations, excitatory and inhibitory: i.e. $n_0 = n_E$ and $n_1 = n_{I}$; $a_0 = \frac{g \mu_E}{\sqrt{N}}$, $a_1 = \frac{-\alpha g \mu_E}{\sqrt{N}}$, $b_0 = b_E a_0$, and $b_1 = b_I a_1$. Then the Jacobian has eigenvalues
\begin{enumerate}
\item $\lambda_E = -1 - \frac{g \mu_E}{\sqrt{N}}(1-b_E)$, with multiplicity $n_E - 1$;
\item $ \lambda_{I} = -1 + \frac{g \alpha \mu_E}{\sqrt{N}}(1-b_I)$, with multiplicity $n_{I} - 1$;
\item 2 remaining eigenvalues given by the $2 \times 2$ matrix $\tilde{\Jvec}$:
\begin{eqnarray}
\tilde{\Jvec} & = & -\Ivec + \frac{g\mu_E}{\sqrt{N}} \left[  \begin{matrix}
n_E - (1-b_E) & -n_E\\
n_E & -n_E + \alpha(1-b_I)
\end{matrix} \right]   \label{eqn:Eproblem}
\end{eqnarray}
This will be a complex pair as long as $n_E > \left( \alpha(1-b_I) + 1-b_E \right)/4$, so $\lambda_{1,2} = \lambda \pm i \omega$ where
\begin{eqnarray*}
\lambda & = & -1 + \frac{g \mu_E}{\sqrt{N}} \frac{\alpha(1-b_I) - 1+b_E}{2}\\
\omega & = & \frac{g \mu_E}{\sqrt{N}} \sqrt{\alpha (1-b_I) + 1-b_E} \sqrt{n_E - \frac{\alpha (1-b_I) + 1-b_E}{4}}
\end{eqnarray*}
We note that $\lambda_E <  \lambda \equiv \Re(\lambda_{1,2}) < \lambda_I$. The eigenvalue associated with the excitatory population, $\lambda_E < 0$ for any value of $g$. 
\end{enumerate}
The corresponding eigenvectors are:
\begin{enumerate}
\item 
$\vvec_E  =  {\rm span} \, \{ \left[  
\vvec_{n_E} \;
\underbrace{\begin{matrix}0 & \cdots & 0\end{matrix}}_{n_{I}} \right] \}, \qquad \vvec_{n_E} \perp \Onevec_{n_E}$;
\item 
$\vvec_{I_1}  =  {\rm span} \, \{ \left[  
\underbrace{\begin{matrix}0 & \cdots & 0\end{matrix}}_{n_{E}} \;
\vvec_{n_{I_1}} \right] \}, \qquad \vvec_{n_{I}} \perp \Onevec_{n_{I}}$;
\item 
$\vvec_{\tilde{J}} =  {\rm span} \, \{ \left[  
\underbrace{\begin{matrix}c_E & \cdots & c_E\end{matrix}}_{n_E} \;
\underbrace{\begin{matrix}c_{I} & \cdots & c_{I}\end{matrix}}_{n_{I}} \right] \}$
\end{enumerate}

We pause to consider two special cases of Example 1. The first is \textit{no} self-coupling --- $b_E, b_I = 0$ --- which we will examine in detail in the rest of this paper. The second is full self-coupling --- $b_E, b_I = 1$ --- which has been studied previously by many authors \cite{RA06,delMolino_etal_PRE_2013,kadmon_HS_2015}. \\

\noindent
\textbf{Example 1.1}:
We consider Example 1, but with no self-coupling: $b_E, b_I = 0$.  Then at the origin $\xvec = 0$, the Jacobian has eigenvalues
\begin{enumerate}
\item $\lambda_E = -1 - \frac{g \mu_E}{\sqrt{N}}$, with multiplicity $n_E - 1$;
\item $ \lambda_{I} = -1 + \frac{g \alpha \mu_E}{\sqrt{N}}$, with multiplicity $n_{I} - 1$;
\item 2 remaining eigenvalues given by the $2 \times 2$ matrix $\tilde{\Jvec}$:
\begin{eqnarray}
\tilde{\Jvec} & = & -\Ivec + \frac{g\mu_E}{\sqrt{N}} \left[  \begin{matrix}
n_E - 1 & -n_E\\
n_E & -n_E + \alpha
\end{matrix} \right]   \label{eqn:Eproblem_1p1}
\end{eqnarray}
which will be a complex pair as long as $n_E > (\alpha + 1)/4$, so $\lambda_{1,2} = \lambda \pm i \omega$ where
\begin{eqnarray*}
\lambda & = & -1 + \frac{g \mu_E}{\sqrt{N}} \frac{\alpha - 1}{2}\\
\omega & = & \frac{g \mu_E}{\sqrt{N}} \sqrt{\alpha + 1} \sqrt{n_E - \frac{1+ \alpha}{4}}
\end{eqnarray*}
We note that $\lambda_E <  \lambda \equiv \Re(\lambda_{1,2}) < \lambda_I$. The eigenvalue associated with the excitatory population, $\lambda_E < 0$ for any value of $g$. In the (un-cortex-like) situation that the excitatory population were \textit{smaller} than the inhibitory population ($\alpha < 1$), then the complex pair would also be stable for all $\lambda < 0$. 
\end{enumerate}
\vspace{0.2in}

\noindent
\textbf{Example 1.2}:
We consider Example 1, but where self-coupling is \textit{not} reduced: $b_E, b_I = 1$.  
Consider the eigenvalues at the origin $\xvec = 0$ described in Example 1:
\begin{enumerate}
\item Since $b_E = 1$, $\lambda_E = -1$ with multiplicity $n_E - 1$;
\item Since $b_I = 1$, $\lambda_{I} = -1$, with multiplicity $n_{I} - 1$;
\item 2 remaining eigenvalues given by the $2 \times 2$ matrix $\tilde{\Jvec}$:
\begin{eqnarray}
\tilde{\Jvec} & = & -\Ivec + \frac{g\mu_E}{\sqrt{N}} \left[  \begin{matrix}
n_E  & -n_E\\
n_E & -n_E 
\end{matrix} \right]   \label{eqn:Eproblem_boring}
\end{eqnarray}
which \textit{also} has the (repeated) eigenvalue $-1$. 
\end{enumerate}
Thus, \textit{every} eigenvalue of $\Hvec$ is $-1$; crucially, this does not depend on the coupling parameter $g$. In \S \ref{sec:solutions_deterministic}, we describe how by varying $g$, bifurcations will occur at the origin; these \textit{cannot occur} if self-coupling is not reduced, as the eigenvalues of the Jacobian cannot pass through the imaginary axis.

(Another way reach the same conclusion, is to notice that $\Hvec$ is a rank-one matrix \cite{RA06}:
\begin{eqnarray}
\Hvec & = & {\uvec} \Onevec^T, \; {\rm where} \; \uvec  =  \left[  
\underbrace{\begin{matrix} \mu_E & \dots & \mu_E \end{matrix}}_{n_E}  \; \underbrace{\begin{matrix} \mu_I & \dots & \mu_I \end{matrix}}_{n_I} 
\right]^T
\end{eqnarray}
with at most one non-zero eigenvalue; since $\mu_E \nu_E + \mu_I \nu_I = 0$, this last eigenvalue is zero as well.)\\

\vspace{0.2in}
\noindent
\textbf{Example 2}:  Next, suppose that the cells have broken into three synchronized populations:  the excitatory cells ($n_E$ cells with activity $x_E(t)$) and two groups of inhibitory cells ($n_{I_1}$ and $n_{I_2}$ cells with activities $x_{I_1}$ and $x_{I_2}$ respectively). Then $n_0 = n_E$, $n_1 = n_{I_1}$ and $n_2 = n_{I_2}$; $a_0 = \frac{g \mu_E}{\sqrt{N}} \sech^2(g x_E)$, and $a_{1,2} = -\frac{g \alpha \mu_E}{\sqrt{N}} \sech^2(g x_{I_{1,2}})$; $b_0 = b_E a_0$ and $b_{1,2} = b_I a_{1,2}$. Therefore the Jacobian has eigenvalues
\begin{enumerate}
\item $\lambda_E = -1 - \frac{g \mu_E}{\sqrt{N}} \sech^2(g x_E) (1-b_E)$, with multiplicity $n_E - 1$;
\item $ \lambda_{I_1} = -1 + \frac{g \alpha \mu_E}{\sqrt{N}} \sech^2(g x_{I_1})(1-b_I)$, with multiplicity $n_{I_1} - 1$;
\item $ \lambda_{I_2} = -1 + \frac{g \alpha \mu_E}{\sqrt{N}} \sech^2(g x_{I_2})(1-b_I)$, with multiplicity $n_{I_2} - 1$;
\item three (3) remaining eigenvalues given by the $3 \times 3$ matrix $\tilde{\Jvec}$ described earlier.
\end{enumerate}
We note that $\lambda_E < 0$ always, as long as $b_E \leq 1$. The corresponding eigenvectors are:
\begin{enumerate}
\item 
$\vvec_E  =  {\rm span} \, \{ \left[  
\vvec_{n_E} \;
\underbrace{\begin{matrix}0 & \cdots & 0\end{matrix}}_{n_{I_1}} \;
\underbrace{\begin{matrix}0 & \cdots & 0 \end{matrix}}_{n_{I_2}} \right] \}, \qquad \vvec_{n_E} \perp \Onevec_{n_E}$;
\item 
$\vvec_{I_1}  =  {\rm span} \, \{ \left[  
\underbrace{\begin{matrix}0 & \cdots & 0\end{matrix}}_{n_{E}} \;
\vvec_{n_{I_1}} \;
\underbrace{\begin{matrix}0 & \cdots & 0 \end{matrix}}_{n_{I_2}} \right] \}, \qquad \vvec_{n_{I_1}} \perp \Onevec_{n_{I_1}}$;
\item
$\vvec_{I_2}  =  {\rm span} \, \{ \left[  
\underbrace{\begin{matrix}0 & \cdots & 0\end{matrix}}_{n_{E}} \;
\underbrace{\begin{matrix}0 & \cdots & 0 \end{matrix}}_{n_{I_1}} \;
\vvec_{n_{I_2}} \;\right] \}, \qquad \vvec_{n_{I_2}} \perp \Onevec_{n_{I_2}}$;
\item 
$\vvec_{\tilde{J}} =  {\rm span} \, \{ \left[  
\underbrace{\begin{matrix}c_E & \cdots & c_E\end{matrix}}_{n_E} \;
\underbrace{\begin{matrix}c_{I_1} & \cdots & c_{I_1}\end{matrix}}_{n_{I_1}} \;
\underbrace{\begin{matrix} c_{I_2} & \cdots & c_{I_2}\end{matrix}}_{n_{I_2}} \right] \}$
\end{enumerate}

\section{Solution families found in deterministic network ($\epsilon = 0$)} \label{sec:solutions_deterministic}
In this section, we use equivariant bifurcation theory to identify which solutions we expect to arise in the system \eqref{eqn:Gsys}, where $\Gvec = \Hvec/\sqrt{N}$. We will also demonstrate that these solutions actually arise in a small network where it is tractable to do numerical continuation to verify our calculations. Our main tool is the \textit{Equivariant Branching Lemma}, which tell us what type of solutions will arise at bifurcation points, when symmetries are present.

Before stating this result, we introduce some terminology: Let $\Gamma$ be a compact Lie group acting on $\mathbb{R}^N$; then we say that a mapping $F: \mathbb{R}^N  \rightarrow \mathbb{R}^N$ is \textit{$\Gamma$-equivariant} if $F(\gamma \xvec) = \gamma F(\xvec)$, for all $\xvec \in \mathbb{R}^N$ and $\gamma \in \Gamma$.  A one-parameter family of mappings  $F: \mathbb{R}^N  \rightarrow \mathbb{R}^N$ is \textit{$\Gamma$-equivariant}, if it is $\Gamma$-equivariant for each value of $\lambda$.

We say that $V$, a subspace of $\mathbb{R}^N$,  is \textit{$\Gamma$-invariant} if $\gamma \vvec \in V$, for any $\vvec$ and $\gamma \in \Gamma$. We furthermore say that the action of $\Gamma$ on $V$ is \textit{irreducible} if $V$ has no proper invariant subspaces; i.e. the only $\Gamma$-invariant subspaces of $V$ are $\{0\}$ and $V$ itself. 

For a group $\Gamma$ and a vector space $V$, we define the \textit{fixed-point subspace} for $\Gamma$, denoted $\rm{Fix} (\Gamma)$, to be all points in $V$ that are unchanged under any of the members of $\Gamma$; i.e. $\rm{Fix} (\Gamma) = \{ \xvec \in V : \gamma \xvec = \xvec, \forall \gamma \in \Gamma \}$. 
The \textit{isotropy subgroup of $\xvec \in V$}, denoted $\Sigma_x$, is the set of all members of $\Gamma$ under which $\xvec$ is fixed; i.e. $\Sigma_x = \{ \gamma \in \Gamma : \gamma \xvec = \xvec \}$. An \textit{isotropy subgroup of $\Gamma$} is a subgroup $\Sigma$ which is the isotropy subgroup, $\Sigma_x$, for some $\xvec \in V$.\\

Suppose we have a one-parameter family of mappings, $F(\xvec, \lambda)$, and we wish to solve $F(\xvec, \lambda)=0$. 
For any $(\xvec, \lambda) \in \mathbb{R}^n \times \mathbb{R}$, let $(dF)_{\xvec,\lambda}$ denote the $N \times N$ Jacobian matrix 
\[ \left( \frac{\partial F_j}{\partial x_k} (\xvec, \lambda) \right)_{j, k=1...N}
\] 
Then the Implicit Function Theorem states that we can continue to track a unique solution branch as a function of $\lambda$, as long as the Jacobian remains invertible. When this ceases to be true --- when $(dF)_{\xvec,\lambda}$ has a nontrivial kernel --- we have the possibility for a bifurcation. At this point the number of zero eigenvalues (whether there are one, or two, etc..) and a menagerie of further conditions, will determine the qualitative properties of the structural change that occurs.

What complicates this situation for $\Gamma$-equivariant mappings --- i.e. $F$ is $\Gamma$-equivariant for any value of the parameter $\lambda$ --- is that because of symmetries, \textit{multiple} eigenvalues will go through zero at once; however, the structural changes that occur are qualitatively the same as those that occur in a non-symmetric system, with a single zero eigenvalue. What changes is that we now have \textit{multiple} such solution branches, each corresponding to a subgroup of the original symmetries.  The following Lemma formalizes this fact:\\

\begin{Theorem} (Equivariant Branching Lemma: paraphrased from \cite{GSS88Vol2}, pg. 82, see also pg. 67-69 ): Let $F: \mathbb{R}^N \times \mathbb{R} \rightarrow \mathbb{R}^N$ be a one-parameter family of $\Gamma$-equivariant mappings with $F(\xvec_0, \lambda_0) = \Zerovec$. Suppose that $(\xvec_0, \lambda_0)$ is a bifurcation point and that, defining $V = \ker(dF)_{\xvec_0,\lambda_0}$, $\Gamma$ acts absolutely irreducibly on $V$. Let $\Sigma$ be an isotropy subgroup of $\Gamma$ satisfying 
\begin{eqnarray}
\rm{dim}\; \rm{Fix} (\Sigma) = 1,
\end{eqnarray}
where $\rm{Fix (\Sigma)}$ is the \emph{fixed-point subspace} of $\Sigma$: that is, $\rm{Fix} (\Sigma) \equiv \{ x \in V \mid \sigma x = x, \;  \forall \sigma \in \Sigma \}$. Then there exists a unique smooth solution branch to $F = 0$ such that the isotropy subgroup of each solution is $\Sigma$. \\
\end{Theorem}
A similar statement holds for Hopf bifurcations, which we state here because we will appeal to its conclusions regarding the symmetry of periodic solutions:
\begin{Theorem} (Equivariant Hopf Theorem: paraphrased from \cite{GSS88Vol2}, pg. 275) Let $F$ be a one-parameter family of $\Gamma$-equivariant mappings with $F(\xvec_0, \lambda_0) = \Zerovec$. Suppose that $(dF)_{\xvec_0,\lambda_0}$ has one or more pairs of complex eigenvalues $\rho \pm i \omega$, which satisfy $\rho(\lambda_0) = 0$ (i.e. the eigenvalues are pure imaginary at $\lambda_0$) and $\rho'(\lambda_0) \not= 0$.
Define $V$ to be the corresponding real (i.e. not generalized) eigenspace. Let $\Sigma$ be an isotropy subgroup of $\Gamma$ satisfying 
\begin{eqnarray}
\rm{dim}\; \rm{Fix} (\Sigma) = 2.
\end{eqnarray}
where $\rm{Fix (\Sigma)}$ is the \emph{fixed-point subspace} of $\Sigma$: that is, $\rm{Fix} (\Sigma) \equiv \{ x \in V \mid \sigma x = x, \;  \forall \sigma \in \Sigma \}$.
Then there exists a unique branch of small-amplitude periodic solutions (with period $2\pi / \omega$), having $\Sigma$ as their group of symmetries.\\
\end{Theorem}
Here, the family of mappings is the right-hand side of Eqn. \eqref{eqn:Gsys}, with $\epsilon = 0$; i.e. $F(\xvec, g) = -\xvec + \Hvec \tanh(g\xvec) /\sqrt{N}$ (our parameter is denoted $g$ rather than $\lambda$). Let $\Gamma = S_{n_E} \oplus S_{n_I}$, where $S_n$ is the symmetric group on $n$ symbols; that is, we are allowed to permute the labels on the excitatory cells, and/or to permute the labels on the inhibitory cells.

It is straightforward to check that $F$ is $\Gamma$-equivariant \footnote{For example, consider $k \le n_E$; then $F_k(\xvec,g) = -x_k - \frac{\mu_E}{\sqrt{N}} \tanh(g x_k) + \sum_{j \le n_E} \frac{\mu_E}{\sqrt{N}} \tanh(g x_j) -\sum_{j > n_E} \frac{\alpha \mu_E}{\sqrt{N}} \tanh(g x_j)  = -x_k - \frac{\mu_E}{\sqrt{N}} \tanh(g x_k) + C$, where $C$ is the same for any cell. $C$ is clearly unchanged under any permutation of the labels of the excitatory cells, or any permutation of the inhibitory cells.}.  Each permutation on $N$ objects can be represented as an element in  $GL(N)$, the group of invertible $N \times N$ matrices; $\Gamma$ is a finite subgroup of such matrices and thus has the structure of a Lie group \cite{GSS88Vol2}; since it has a finite number of elements it is also bounded and thus compact. 
 
Since our model satisfies the assumptions of the Equivariant Branching Lemma, it remains for us to identify potential bifurcation points (we concentrate on absent self-coupling, i.e. $b_E, b_I = 0$). From the trivial solution ($\xvec = \Zerovec$), we expect solutions to arise when the eigenvalues of $-\Ivec + g \Hvec/\sqrt{N}$ cross the imaginary axis. In particular, we expect, in order of increasing $g$,
\begin{itemize}
\item A branch of fixed-point solutions when $g^{\ast} = \sqrt{N}/\alpha/\mu_E$: where the eigenvalues corresponding to the inhibitory population cross zero: here the $I$ cells break into 2 groups of size $n_{I1}$ and $n_{I2}$. Along this fixed point branch, the two groups remain clustered; the excitatory cells also remain clustered, i.e. the solution branch can be characterized by $(x_E,x_{I1},x_{I2})$.  We refer to this as the ``$I_1/I_2$ branch".
\item A branch of limit cycles emerging from a Hopf bifurcation when $g = 2\sqrt{N}/ \mu_E/ (\alpha-1)$: here a complex pair cross the imaginary axis.\\
\end{itemize}

From each $I_1/I_2$ branch, we find:\\
\begin{itemize}
\item A branch of limit cycles from a Hopf bifurcation (at $g^H$) in which the three cluster pattern is maintained: i.e. activity can be characterized by $(x_E, x_{I1}, x_{I2})$.
\item If $n_{I1} = n_{I2}$, then the excitatory activity along this branch is zero: there may be a further branch point, in which $x_E$ moves away from the origin, while $I$ cells remain in their distinct clusters.
\item  (Possibly) other fixed point branches, in which one inhibitory cluster ($x_{I1}$) breaks into further clusters.
\end{itemize}

\subsection{Branch of fixed points (from trivial solution)}

The first opportunity for a bifurcation from the trivial solution occurs when $g^{\ast} = \sqrt{N}/\alpha/\mu_E$: at this value of $g$, $n_I -1$ eigenvalues pass through zero: the corresponding eigenspace (from Example 1) is the set of all zero-sum vectors with support in the inhibitory cells only; i.e. \[ V \equiv  \ker(dF)_{\Zerovec,g^{\ast}}  = {\rm span} \, \{ \left[  
\underbrace{\begin{matrix}0 & \cdots & 0\end{matrix}}_{n_{E}} \;
\vvec_{n_{I}} \right] \}, \qquad \vvec_{n_{I}} \perp \Onevec_{n_{I}}.\]
To check that $\Gamma$ acts irreducibly on $V$ it is sufficient to show that the subspace spanned by the \textit{orbit} of a single vector $\vvec$ (defined as the set of all values  $\gamma \vvec$, for all $\gamma \in \Gamma$) is full rank; this can be readily confirmed for $\vvec_{n_{I}} = \left[ \begin{array}{ccccc} 1 & -1 & 0 & ... & 0 \end{array} \right]$, for example.   

Suppose we break the inhibitory cells up into precisely two clusters; we allow all permutations within each cluster, but no longer allow mixing between the clusters. This describes a subgroup of $\Gamma$, $\Sigma = S_{n_E} \oplus S_{n_{I_1}} \oplus S_{n_{I_2}}$, $n_{I_1} + n_{I_2} = n_I$.
Assuming that (without loss of generality) the $I_1$ neurons have the indices $n_E+1,...,n_E+n_{I_1}$, and so forth, $\Sigma$ has the fixed point subspace 
\begin{eqnarray}
\rm{Fix}(\Sigma) & = & {\rm span} \, \{ \left[  
\underbrace{\begin{matrix}0 & \cdots & 0\end{matrix}}_{n_E} \;
\underbrace{\begin{matrix}1 & \cdots & 1\end{matrix}}_{n_{I_1}} \;
\underbrace{\begin{matrix}-\frac{n_{I_1}}{n_{I_2}} & \cdots & -\frac{n_{I_1}}{n_{I_2}} \end{matrix}}_{n_{I_2}} \right] \}
\end{eqnarray}
We can check that $\rm{Fix}(\Sigma)$ is a subspace of $V$; furthermore $\dim \rm{Fix}(\Sigma) = 1$ because it can be described as the span of a single vector.

Thus, the Equivariant Branching Lemma tells use that we can expect a new branch of fixed points in which the inhibitory cells break up into two groups (therefore we refer to this as the ``$I_1/I_2$ branch"). 

 
If the clusters are of equal size ($n_{I_1} = n_{I_2}$), then the solution branch shows the pattern $(0,x_{I_1},-x_{I_1})$ (by uniqueness, it suffices to show that such a branch exists). To see this, first observe that 
\[
\frac{dx_E}{dt} = -x_E - \frac{\mu_E}{\sqrt{N}}\tanh(g x_E) + C,  \qquad \frac{dx_{I_{1,2}}}{dt} = -x_{I_{1,2}} + \frac{\alpha \mu_E}{\sqrt{N}}\tanh(g x_{I_{1,2}}) + C
\]
where 
\[
C  = \frac{\sqrt{N} \mu_E \alpha}{\alpha+1} \left( \tanh(g x_E) -(1/2) \tanh(g x_{I_1}) - (1/2) \tanh(g x_{I_2}) \right).
\]
If $x_{I_2} = -x_{I_1}$, then $\tanh(g x_{I_2}) = -\tanh(g x_{I_I})$ and therefore
\begin{eqnarray*}
\frac{dx_{I_{1}}}{dt}  + \frac{dx_{I_{2}}}{dt}  & = & -x_{I_{1}} -x_{I_{2}} + \frac{\alpha \mu_E}{\sqrt{N}} \left( \tanh(g x_{I_{1}}) + \tanh(g x_{I_{2}}) \right) + 2C\\
& = & -x_{I_{1}} +x_{I_{1}} + \frac{\alpha \mu_E}{\sqrt{N}} \left( \tanh(g x_{I_{1}}) - \tanh(g x_{I_{1}}) \right) + 2C = 2C 
\end{eqnarray*}
while
\begin{eqnarray*} 
C &  = & \frac{\sqrt{N} \mu_E \alpha}{\alpha+1} \left( \tanh(g x_E) -(1/2) \tanh(g x_{I_1}) + (1/2) \tanh(g x_{I_1}) \right) = \frac{\sqrt{N} \mu_E \alpha}{\alpha+1}  \tanh(g x_E) 
\end{eqnarray*}
Since $\frac{dx_{I_{1}}}{dt}  + \frac{dx_{I_{2}}}{dt} = 0$, $\tanh(g x_E) = 0 \Rightarrow x_E = 0$.

Returning to the inhibitory degrees of freedom, we see their equations are now decoupled:
\[
\frac{dx_{I_{1,2}}}{dt} = -x_{I_{1,2}} + \frac{\alpha \mu_E}{\sqrt{N}}\tanh(g x_{I_{1,2}}) \]
a fixed point has \textit{three} possible solutions, if $g > g^{\ast}$; one is $x_{I_{1,2}}  = 0$, while the others can be found by inverting a simple expression relating
$g$ and $x_{I_1}$ along the solution branch:
\begin{eqnarray}
 -x_{I_{1}} + \frac{\alpha \mu_E}{\sqrt{N}}\tanh(g x_{I_{1}}) = 0 & \Rightarrow & g = \frac{1}{x_{I_1}} \tanh^{-1} \left( \frac{\sqrt{N}x_{I_1}}{\alpha \mu_E}\right) \label{eqn:g_tanhinv}
\end{eqnarray}
Thus, we can solve for $g$ as a function of $x_{I_1}>0$ and set $x_{I_{2}} = -x_{I_{1}}$; checking the Taylor expansion of Eqn. \eqref{eqn:g_tanhinv} will confirm that $x_{I_{1}} \rightarrow 0$ as $g\rightarrow g^{\ast}$. 

%

\subsection{Hopf bifurcation (on trivial solution) leading to limit cycles} \label{sec:Hopf_trivial}
The trivial solution is next expected to have a bifurcation when the complex pair of eigenvalues of $-\Ivec + g \Hvec/\sqrt{N}$ crosses the imaginary axis: that is, when
\begin{eqnarray*}
g & = & \frac{2 \sqrt{N}}{\mu_E (\alpha-1)}
\end{eqnarray*} 
This is a simple eigenvalue pair, with real eigenspace (again by Example 1) consisting of vectors with all E cells synchronized and all I cells synchronized. This is a two-dimensional vector space: therefore, we expect a branch of periodic solutions to arise in which the excitatory neurons and inhibitory neurons are each synchronized.  Here $\Sigma = \Gamma = S_{n_E} \oplus S_{n_I}$.

\subsection{Hopf bifurcation (on $I_1/I_2$ branch) leading to limit cycles} 
On the branch $(x_E, x_{I_1}, x_{I_2})$ we find two singularities that lead to new structures. Most significantly we find a supercritical Hopf bifurcation that leads to a branch of limit cycles, when a pair of complex eigenvalues crosses the imaginary axis. By Example 2, the corresponding eigenspace is fixed under  $\Sigma = S_{n_E} \oplus S_{n_{I_1}} \oplus S_{n_{I_2}}$.  Thus, it is a two-dimensional subspace of $\rm{Fix}(\Sigma)$; therefore, by the Equivariant Hopf Theorem, the family of periodic solutions that emerges here also has $\Sigma$ as its group of symmetries \footnote{While we do not need this theorem to tell us that a Hopf bifurcation occurs, as the eigenvalue pair is simple, it does guarantee that the resulting solutions have the same symmetry group}. 

%
In general 
it is not feasible to solve for $g^{H}$ symbolically: this requires us to solve for the roots of a cubic polynomial involving exponential functions (e.g. $\tanh(g x_E)$) of implicitly defined parameters $x_E$, $x_{I_1}$, and $x_{I_2}$. However, we can identify the bifurcation numerically (all continuations were performed with MATCONT \cite{Matcont03}), and we have found this bifurcation on every specific $I_1/I_2$ branch in every specific system we have investigated.

We can also track the branch of Hopf points numerically in the reduced system $(x_E, x_{I_1}, x_{I_2})$ (described in \S \ref{sec:3cluster}), which has the added benefit that the complexity of the system does not increase with $N$ (rather $N$ is a bifurcation parameter).  Here again, we can confirm that the Hopf bifurcation is present in the system for any $N$, and have done so for several example $n_{I_1}/n_{I_2}$ ratios in \S \ref{sec:largeN}.

\subsection{Branch points (on $I_1/I_2$ branch) leading to new fixed point branch}
We may also find branch points on the $(x_E, x_{I_1}, x_{I_2})$
curve, in which one of the inhibitory clusters breaks into a further cluster. This will occur if the eigenspace corresponding to $x_{I_1}$, say, has a real eigenvalue going through zero. Because these did not appear to play a significant role in our simulations, we will consider them no further here. 

\subsection{Reduced self-coupling}
For the remainder of the paper, we will focus on absent self-coupling ($b_E, b_I = 0$); here we note how our conclusions would be modified, in the more general case. At the origin, the locations --- but not the qualitative behavior --- of the bifurcations will change. In Example 1.1, a branch point occurs at $g^{\ast} = \frac{\sqrt{N}}{\alpha \mu_E}$; in Example 1, the location is now $g^{\ast, b} = \frac{\sqrt{N}}{\alpha \mu_E (1-b_I)}$. Since $b_I \le 1$,  $g^{\ast,b} \ge g^{\ast}$ always. 

Similarly the Hopf bifurcation which occurs at
\begin{eqnarray*}
g^H & = & \frac{2 \sqrt{N}}{\mu_E (\alpha-1)}
\end{eqnarray*} 
with no self-coupling will now occur at
\begin{eqnarray*}
g^{H,b} & = & \frac{2 \sqrt{N}}{\mu_E \left( \alpha (1-b_I) - (1-b_E)\right)}
\end{eqnarray*} 
provided that $\alpha(1-b_I) - (1-b_E) > 0$ (see the formula for $\lambda$ below Eqn. \eqref{eqn:Eproblem}).

The relative ordering of $g^H$ and $g^{H,b}$ would depend on the relative values of $b_E$ and $b_I$; if $b_E - \alpha b_I \le 0$, then $g^{H,b} \ge g^{H}$; otherwise $g^{H,b} < g^{H}$. However, we can check that the branch point will almost always occur for a smaller coupling value (than the Hopf point); that is $g^{\ast,b} \le g^{H,b}$, with equality if and only if $b_E = 1$.

\subsection{Inhibition-dominated networks}
In this paper we have focused on balanced networks ($\alpha = n_I/n_E$). We briefly summarize how our conclusions would change, in inhibition-dominated networks ($\alpha > \tilde{\alpha} \equiv n_I/n_E$). At the origin, the location of $g^{\ast}$ would still be given by $\frac{\sqrt{N}}{\alpha \mu_E (1-b_I)}$, although now 
since $\alpha > \tilde{\alpha}$, the critical coupling value would decrease; i.e. $g^{\ast,b,in} < g^{\ast,b}$. 

In Eqns. \eqref{eqn:Eproblem} and \eqref{eqn:Eproblem_1p1}, the condition that $n_E = \alpha n_I$ has been used; to remove this restriction, replace any instance of $n_E$ in the right column of $\tilde{\Jvec}$ with $\alpha n_I$. 
The condition for a Hopf bifurcation to occur at the origin would now be (using the trace of $\tilde{\Jvec}$ from Eq. \ref{eqn:Eproblem}):
\[ n_E - \alpha n_I + \alpha(1-b_I) - (1-b_E)  > 0 \Rightarrow  (\tilde{\alpha}-\alpha) n_I + \alpha(1-b_I) - (1-b_E)  > 0\]
or
\[ \alpha (1-b_I) - (1-b_E) > (\alpha - \tilde{\alpha}) (1-f) N \]
Thus the Hopf bifurcation will still occur as long as inhibition is not too strong (as measured by $\alpha - \tilde{\alpha}$); however, this depends on $N$.

\section{A bifurcation-preserving reduced-order model}  \label{sec:N20}

In this section, we show that we can construct a reduced-order model that preserves the dynamics and bifurcation structure of the full system, but with a dramatic reduction in the number of degrees of freedom. For a cortex-like ratio of E to I cells, the interesting bifurcations occur surrounding the eigenvalues associated with the inhibitory cells or the complex pair.  As a result, all the ``action" is in the I cells, with the E cells perfectly synchronized always. In fact, we can formalize this as follows:\\

\begin{Lemma}Any fixed point or periodic solution of \eqref{eqn:Gsys}, \eqref{eqn:G_def} with $\epsilon = 0$ has a synchronized excitatory population: i.e. $x_j(t) = x_k(t)$, for any $j,k \leq n_E$.\\
\end{Lemma}
\begin{proof}
Consider the activity of two distinct $E$ cells, say $x_1$ and $x_2$. Then 
\begin{eqnarray}
\frac{d(x_1 - x_2)}{dt} & = & \dot{x_1} - \dot{x_2} \nonumber \\
& = & -x_1 + \frac{\mu_E}{\sqrt{N}} \tanh(g x_2) - (-x_2 + \frac{\mu_E}{\sqrt{N}} \tanh(g x_1)) \label{eqn:x1_minus_x2_A}\\ 
& = & -(x_1 - x_2) - \frac{\mu_E}{\sqrt{N}} \left( \tanh(g x_1) - \tanh(g x_2) \right) \nonumber \\
& = & -(x_1 - x_2) - \frac{\mu_E}{\sqrt{N}} \tanh \left( g(x_1 - x_2) \right) \Bigl( 1 - \tanh(g x_1) \tanh(g x_2) \Bigr) \label{eqn:x1_minus_x2_B}
\end{eqnarray}
The first numbered line, \eqref{eqn:x1_minus_x2_A}, contains so few terms because everything depending on other variables ($x_3$, and so forth) cancels out; the second line, \eqref{eqn:x1_minus_x2_B}, uses a sum identity for the tanh function.  
Then 
\begin{eqnarray*}
\frac{d \| x_1 - x_2 \|^2}{dt} & = & 2(x_1 - x_2) \times  \Bigl[ -(x_1 - x_2) - \frac{\mu_E}{\sqrt{N}} \tanh \left( g(x_1 - x_2) \right) \Bigl( 1 - \tanh(g x_1) \tanh(g x_2) \Bigr)  \Bigr]\\
& = & -2\| x_1 - x_2 \|^2 -  \frac{\mu_E}{\sqrt{N}} \Bigl( (x_1-x_2) \tanh \left(g(x_1-x_2 \right) \Bigr) \Bigl( 1 - \tanh(g x_1) \tanh(g x_2) \Bigr)\\
& \leq & -2\| x_1 - x_2 \|^2
\end{eqnarray*}
with equality if and only if $x_1 = x_2$. In the last line, we use the facts that $x \tanh (gx) \geq 0$ and $( 1 - \tanh(g x) \tanh(g y)) \geq 0$ for any real numbers $x$ and $y$, and $g > 0$.
Therefore the distance $\| x_1 - x_2 \|$ will always decrease along a trajectory, unless $x_1 = x_2$ already.
\end{proof}
As a consequence, any fixed point or period solution present in the full system is also present in the following reduced system, where we collapse all of the excitatory degrees of freedom into one $x_E$:
\begin{eqnarray}
\dot x_E & = & -x_E + \left( \frac{N\alpha}{\alpha+1} - 1\right) \left( \frac{\mu_E}{\sqrt{N}}\right) \tanh(g x_E) - \sum_{j=1}^{n_I} \left( \frac{\alpha \mu_E}{\sqrt{N}}\right) \tanh(g x_{I_j})  \label{eqn:reduced_plus1_xE}\\
\dot x_{I_i} & = & -x_{I_i} + \frac{N\alpha}{\alpha+1}    \left( \frac{\mu_E}{\sqrt{N}}\right)  \tanh(g x_E) - \sum_{j=1,j\not=i}^{n_I} \left( \frac{\alpha \mu_E}{\sqrt{N}}\right) \tanh(g x_{I_j}), \qquad i=1,...,n_I  \label{eqn:reduced_plus1_xI}
\end{eqnarray}
Here, $n_I$ is fixed, while $N = (\alpha + 1)n_I$ is a parameter of the system. This allows us to explore solutions in a $n_I + 1$ dimensional system rather than a $(\alpha+1) n_I$ dimensional system. 

\begin{figure}[t]
\begin{center}
\includegraphics[width=5cm]{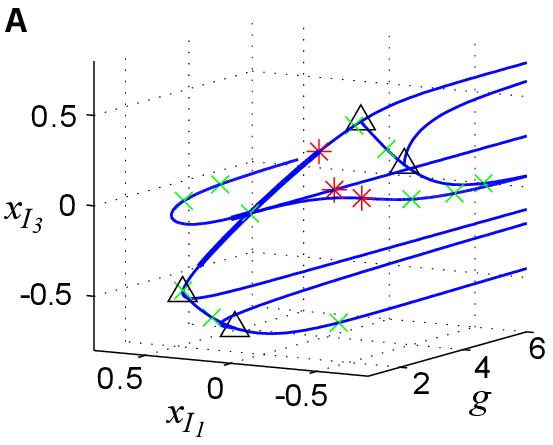}
\includegraphics[width=5cm]{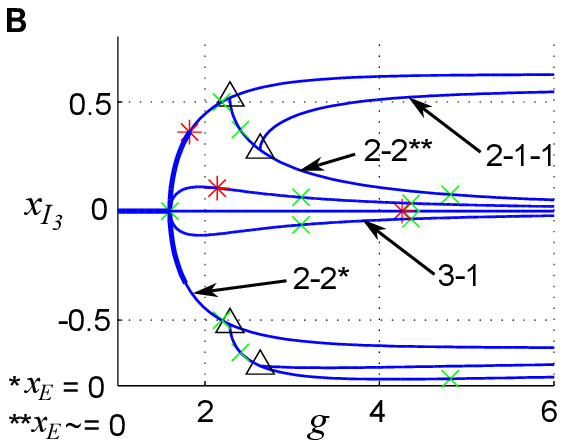}
\includegraphics[width=5cm]{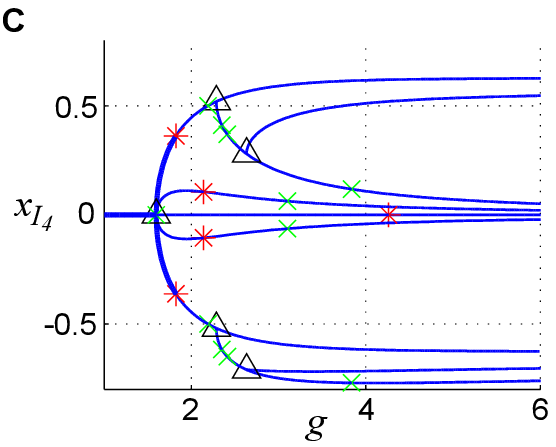}
\caption{Solution branches from symmetry, in the $N=20$ system. (A) Solution branches (up to symmetry) in $(g, x_{I_1}, x_{I_3})$ coordinates. (B) Same figure, viewed in the $(g, x_{I_3})$ plane. (C) Solution branches from the reduced $n_I + 1$-dimensional system, with $n_I = 4$ and $\alpha = 4$. Up to symmetry, this figure depicts identical solutions as the left and center panels. Markers indicate: Hopf bifurcations (red asterisks); branch points (black triangles); neutral saddles (green crosses). } \label{fig:eq_sym_sol_branch}
\end{center}
\end{figure}

\begin{figure}[t]
\begin{center}
\includegraphics[width=5cm]{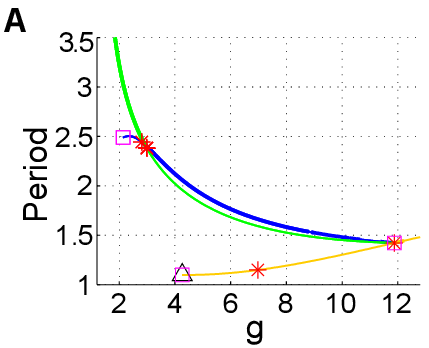} \;
\includegraphics[width=5cm]{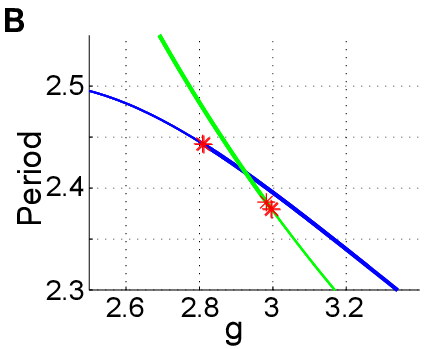}
\end{center}
\caption{Limit cycles in the $N=20$ system. (A) Period vs coupling parameter $g$. (B) Closeup of period vs. $g$, near a stability change in the 3-1 curve. Colors are: 3-1 (blue), 2-2 (green), E/I (orange). Markers indicate: Neiman-Sackler (red asterisk), limit point of cycles (magenta square), branch pt. of cycles (black triangle)  } \label{fig:lc_sym_sol_branch}
\end{figure}

\begin{figure}[t]
\begin{center}
\includegraphics[width=5cm]{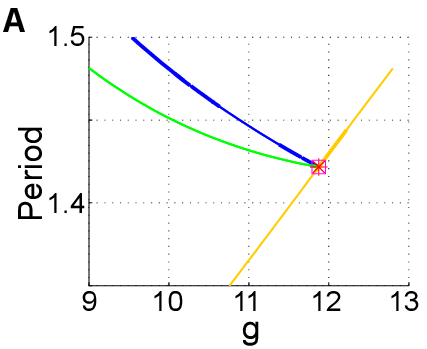} 
\includegraphics[width=5cm]{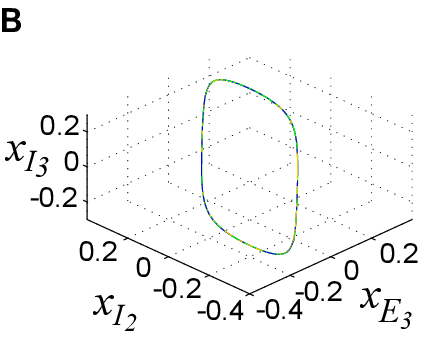}
\end{center}
\caption{Limit cycles in the $N=20$ system. (A) Closeup of period vs. $g$, near the point where three limit cycle branches collide. (B) Limit cycles at the point of collision, in $(x_E, x_{I_2}, x_{I_3})$ coordinates. Colors are: 3-1 (blue), 2-2 (green), E/I (orange). Markers indicate: Neiman-Sackler (red asterisk), limit point of cycles (magenta square), branch pt. of cycles (black triangle)  } \label{fig:lc_sym_sol_zoom_bif}
\end{figure}

We first demonstrate the bifurcation structure on a small network, $N=20$, in which we can comfortably confirm our findings on the full system with numerical continuation (all diagrams shown here were computed using MATCONT \cite{Matcont03}).
We treat the global coupling strength, $g$, as our bifurcation parameter: the origin is an equilibrium point for all $g$. 
At $g^{\ast} = \sqrt{N}/\alpha/\mu_E$, $n_I - 1$ eigenvalues pass through the origin ($\epsilon = 0$). This is a branch point: because of symmetry, there exists a branch corresponding to each possible split of the $I$ cells into two clusters. In Fig. \ref{fig:eq_sym_sol_branch}A, we show the solution branches that arise in the $N=20$ system (up to symmetry, that is: while there are four possible 3-1 splittings, we display only one here). Because there are four inhibitory cells, there are two possible splits: 3-1 and 2-2. Both have a branch that originates from the branch point on the origin $g^{\ast} = \sqrt{N}/\alpha/\mu_E$ (in Fig. \ref{fig:eq_sym_sol_branch}B, these are labeled as ``3-1" and ``2-2, e=0" respectively). Along the 2-2 branch, the E cells have zero activity (this is generally the case when the I cells split into two equal clusters). Both branches then have a Hopf bifurcation from which a branch of limit cycles emerges; unstable in the 3-1 case, stable in the 2-2 case. 
The resulting limit cycle respects the clustering, but the E cell activity is no longer zero in the 2-2 case.  

The 2-2 branch has a further branch point, at which the new branch breaks the $i_1/i_2$ odd symmetry and E cell activity moves away from zero.
One further branch occurs, in which one of the 2 cell clusters breaks apart resulting in a 2-1-1 clustering.
Why did the 2-1-1 branch come off of the 2-2 branch, rather than the 3-1 branch? At this time, we don't have a principled answer. 
Finally, the origin has a Hopf bifurcation in which the E cells and I cells separately cluster (we will refer to this as the ``E/I" limit cycle). 


We next perform the same continuation on the corresponding reduced (5-dimensional: $x_E$ and $x_{I_1}-x_{I_4}$) system. The equilibrium branch structure is shown in Fig. \ref{fig:eq_sym_sol_branch}C. Up to a permutation of the inhibitory coordinate labels (we did not force the same cell cluster identities to be tracked in both continuations), the curves are identical. 

Returning to the full system, we now consider the limit cycles which emerge from the three Hopf bifurcations we identified (on the 3-1 branch, 2-2 branch, and the origin).  
In Fig. \ref{fig:lc_sym_sol_branch}A, we plot the period vs. the coupling parameter $g$. In Fig. \ref{fig:lc_sym_sol_branch}B we can see that the 3-1 branch is stable for $g > \approx 2.8$; the 2-2 branch, for $g < \approx 3$. We note that the 3-1 and 2-2 branches appear to terminate on the E/I branch, shown in Fig. \ref{fig:lc_sym_sol_zoom_bif}A. Indeed, at this point all three limit cycles coincide, as we can see in  Fig. \ref{fig:lc_sym_sol_zoom_bif}B.


\subsection{A larger system: $n_I = 10$}
\begin{figure}[t]
\begin{center}
\includegraphics[height=5.2cm]{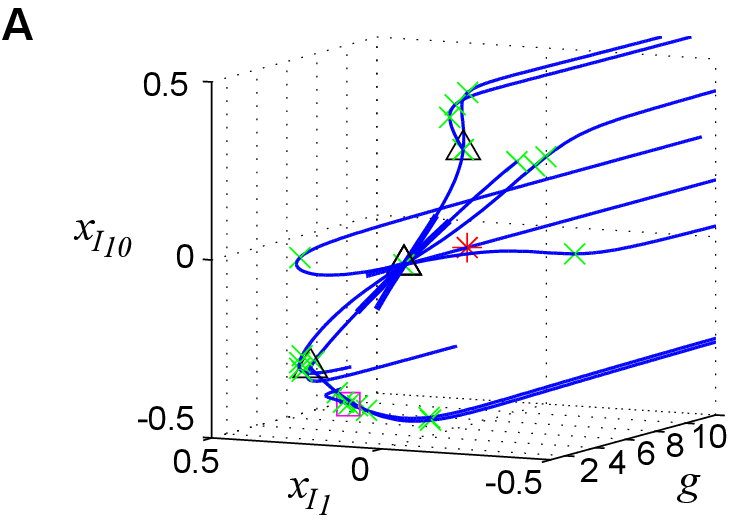}\; \; 
\includegraphics[height=5.2cm]{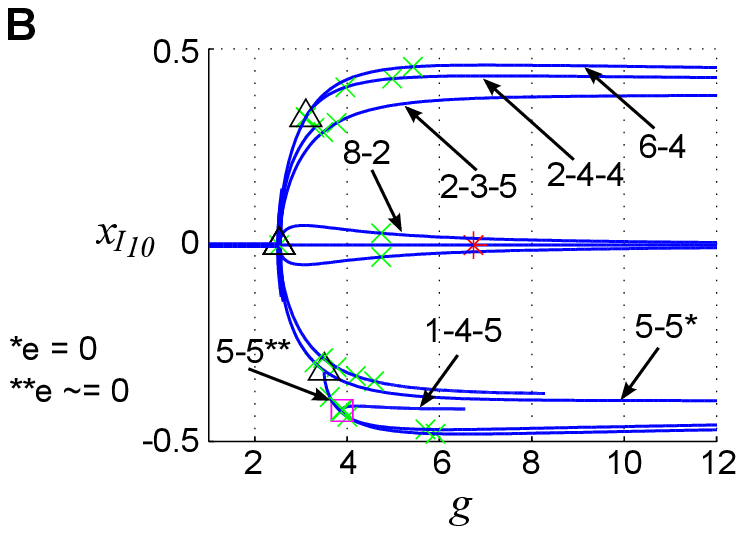}
\caption{Equilibrium solution branches in the  ``$x_I+1$" system, for $n_I = 10$. With the parameters chosen, this is equivalent to a $N=50$ full system. Two different viewpoints are shown (panels A and B). Markers indicate: Hopf bifurcations (red asterisks); branch points (black triangles); neutral saddles (green crosses).} \label{fig:eq_xIplus1}
\end{center}
\end{figure}
The real power of the reduced-order model becomes evident when we increase population size. We now show results for $n_I = 10$: note that for $\alpha=4$, we reduce the dimensionality of the system from 50 to 11. In Fig. \ref{fig:eq_xIplus1}A we show the equilibrium branches found in this system; the same diagram is plotted in the $(g, x_{I_{10}})$ plane, with labeled curves, in Fig. \ref{fig:eq_xIplus1}B. As expected, a branch point occurs on the origin at $g = \sqrt{N}/\alpha \mu_E$. From this point we see 5-5, 6-4, 8-2, and 2-3-5 solutions emerge. The 5-5 branch has zero activity in the excitatory cells; as in the previous example, a further branching point yields 5-5 solutions where $x_E \not= 0$. A further branch point gives a curve of 1-4-5 cluster solutions. There are Hopf bifurcations on each of the branches that appear at the origin.

We note that \textit{most} of these branches are cases where splitting is minimal; that is, a single cluster breaks into two (rather than into three). This confirms our intuition from the Equivariant Branching Lemma, which guarantees the existence of a unique branch of solutions for each subgroup $\Sigma$ for which the fixed point subspace on the kernel of the Jacobian at the bifurcation point has the right dimension: $\dim \Fix (\Sigma) = 1$. (A more general result
extends a version of this result to cases of odd dimensions \cite{Lauterbach2010}.)  At the origin, for example, the kernel at the branch point is $n_I-1$ dimensional:
\begin{eqnarray}
\vvec & = & \left[ \begin{matrix} 0 & \uvec \end{matrix} \right], \qquad \uvec \in \R^{n_I}, \;  \uvec^T \mathbf{1} = 0 
\end{eqnarray}
In this case, 
\begin{eqnarray}
\vvec & = & \left[ \begin{matrix} 0 & v_1 & \cdots & v_{10} \end{matrix} \right], \qquad  v_1 + \cdots + v_{10} = 0 
\end{eqnarray}
However, this lemma does not exclude the possibility of other solution types, and little is known in general about fixed point subspaces of \textit{even} dimensions: such solutions have been found in some systems (for example, \cite{Lauterbach2010}), but there is currently not a general theory guaranteeing or ruling out such solutions \cite{Lauterbach2015}. In this system, at least one branch corresponds to a subgroup $\Sigma$ for which $\dim \Fix (\Sigma) = 2$: the 2-3-5 branch.
\begin{figure}[t]
\begin{center}
\includegraphics[height=5cm]{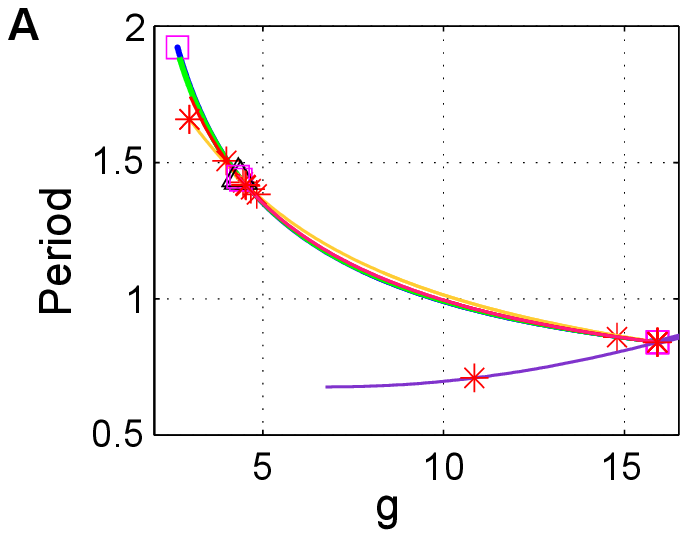}\; \; 
\includegraphics[height=5cm]{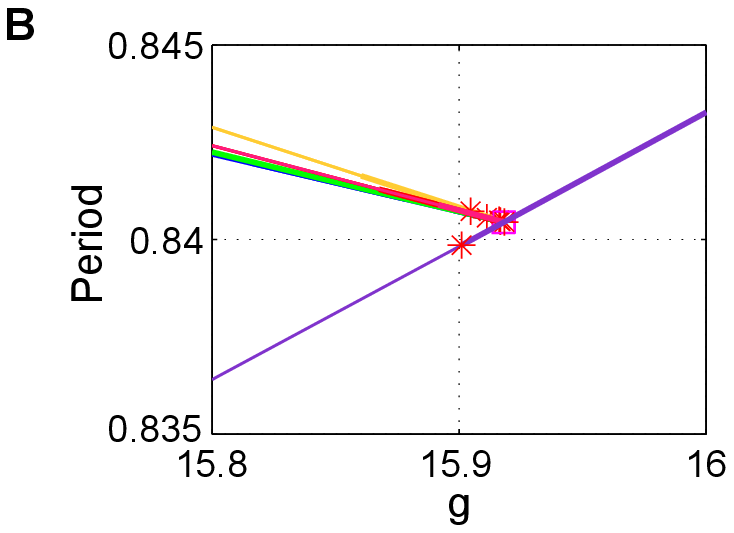}
\caption{Limit cycle branches in the  ``$x_I+1$" system. (A) Period vs coupling parameter $g$.  (B) Closeup of period vs. $g$, near the point where several limit cycle branches collide. Colors are: 5-5 (blue), 6-4 (green), 8-2 (orange), 2-3-5 (red), 2-3-5 secondary (pink), E/I (purple). Markers indicate: Neiman-Sackler (red asterisk), limit point of cycles (magenta square), branch pt. of cycles (black triangle)} \label{fig:lc_xIplus1}
\end{center}
\end{figure}

We next look at the limit cycles that arise on from Hopf bifurcations on each branch from the origin. Period decreases with $g$ (Fig. \ref{fig:lc_xIplus1}A). As in the $N=20$ system, each branch terminates where it collides with the E/I limit cycle branch that comes off the Hopf point at the origin (Fig. \ref{fig:lc_xIplus1}B).

%

\subsection{Reduced system: 3-cluster $(x_E, x_{I_1}, x_{I_2})$}  \label{sec:3cluster}
We can gain additional insight into arbitrarily large systems by reducing \eqref{eqn:Gsys} using the assumption of a three cluster grouping into populations of $n_E$, $n_{I_1}$, and $n_{I_2}$, whose activities are given by $x_E$, $x_{I_1}$ and $x_{I_2}$ respectively. The reduced system is

\begin{eqnarray*}
\dot x_E & = & -x_E + (n_E - 1) (\mu_E/\sqrt{N}) \tanh(g x_E) + n_{I_1} (\mu_I/\sqrt{N}) \tanh(g x_{I_1}) +  n_{I_2} (\mu_I/\sqrt{N}) \tanh(g x_{I_2}) \label{eqn:red_xE}\\
\dot x_{I_1} & = & -x_{I_1} + n_E  (\mu_E/\sqrt{N}) \tanh(g x_E) + (n_{I_1} - 1) (\mu_I/\sqrt{N})  \tanh(g x_I) +   n_{I_2} (\mu_I/\sqrt{N}) \tanh(g x_{I_2})\\
\dot x_{I_2} & = & -x_I + n_E  (\mu_E/\sqrt{N}) \tanh(g x_E) +   n_{I_1} (\mu_I/\sqrt{N}) \tanh(g x_{I_1})+ (n_{I_2} - 1) (\mu_I/\sqrt{N})  \tanh(g x_{I_2}). \label{eqn:red_xI_2}
\end{eqnarray*}
We can also parameterize the clustering with $\alpha$ and $\beta$ such that $n_E = \frac{\alpha}{\alpha + 1}N$, $n_{I_1} = \frac{\beta}{\beta+1}\frac{\alpha}{\alpha + 1}N$, and $n_{I_2} = \frac{1}{\beta+1}\frac{1}{\alpha + 1}N$; that is, $\beta$ gives the ratio of $n_{I_1}$ to $n_{I_2}$, just as $\alpha$ gives the ratio of $n_E$ to $n_I$. Then the equations become (also using the relationship $\mu_I = -\alpha \mu_E$):
\begin{eqnarray}
\dot x_E & = & -x_E + \left( \frac{N\alpha}{\alpha+1} - 1\right) \left( \frac{\mu_E}{\sqrt{N}}\right) \tanh(g x_E) - \frac{N\beta }{(\beta+1)(\alpha+1)} \left( \frac{\alpha \mu_E}{\sqrt{N}}\right) \tanh(g x_{I_1})   \nonumber \\
&&- \frac{N}{(\beta+1)(\alpha+1)} \left(\frac{\alpha \mu_E}{\sqrt{N}}\right) \tanh(g x_{I_2}) \label{eqn:reduced_3cluster_xE}\\
\dot x_{I_1} & = & -x_{I_1} + \frac{N\alpha}{\alpha+1}    \left( \frac{\mu_E}{\sqrt{N}}\right)  \tanh(g x_E) - \left(\frac{N\beta}{(\beta+1)(\alpha+1)}-1 \right) \left( \frac{\alpha \mu_E}{\sqrt{N}}\right)  \tanh(g x_{I_1}) \nonumber \\
&& - \frac{N}{(\beta+1)(\alpha+1)}  \left( \frac{\alpha \mu_E}{\sqrt{N}}\right)    \tanh(g x_{I_2})\\
\dot x_{I_2} & = & -x_{I_2} + \frac{N\alpha}{\alpha+1} \left( \frac{\mu_E}{\sqrt{N}}\right)    \tanh(g x_E) -   \frac{N\beta}{(\beta+1)(\alpha+1)} \left( \frac{\alpha \mu_E}{\sqrt{N}}\right)   \tanh(g x_{I_1}) \nonumber \\
&& - \left(\frac{N }{(\beta+1)(\alpha+1)}-1 \right) \left( \frac{\alpha \mu_E}{\sqrt{N}}\right)   \tanh(g x_{I_2}); \label{eqn:reduced_3cluster_v2}
\end{eqnarray}

Here, we can treat $N$, $\alpha$ and $\beta$ as continuously varying bifurcation parameters. When $N$, $\frac{N}{\alpha+1}$, and $\frac{N}{(\beta+1)(\alpha+1)}$ are all positive integers, the reduced system \eqref{eqn:reduced_3cluster_xE}-\eqref{eqn:reduced_3cluster_v2} lifts onto an $N$-cell network.

\subsection{Scaling with system size} \label{sec:largeN}
We can use this reduced system to explore how the system behaves as $N$ increases. The system in Eqn. \eqref{eqn:reduced_3cluster_xE}-\eqref{eqn:reduced_3cluster_v2} allows $N$ to be a continuously varying parameter; therefore, we can vary $N$ while holding all other parameters fixed. Notably, we will keep $\beta$ fixed; thus, we will track the behavior of a specific partition \textit{ratio} of inhibitory cells (such as 1-to-1 or 3-to-1), as $N$ increases.
When $N$, $\frac{N}{\alpha+1}$, and $\frac{N}{(\beta+1)(\alpha+1)}$ are all positive integers, the reduced system lifts onto an $N$-cell network; at each such $N$, we can track the $I_1/I_2$ fixed point branch from the known bifurcation point $g^{\ast} = \sqrt{N}/\alpha/\mu_E$. 

In Figure \ref{fig:bif_n1_2_n2_2}A, we show $(x_E, x_{I_I}, x_{I_2})$ as a function of $g$ for the partition $n_{I_1} = n_{I_2}$. Colors cycle through $N$; for each $N$, the curves from top to bottom  indicate $x_{I_1}$, $x_E$, and $x_{I_2}$. We can also locate the Hopf bifurcation along this branch, at $g^{H}$, and measure the frequency of the periodic solutions that emerge at that point. We plot these quantities in Figure \ref{fig:bif_n1_2_n2_2}B: we can see they each scale like $\sqrt{N}$.  In Figure \ref{fig:bif_n1_1_n2_4}, we show the same quantities computed for two more examples: 1-to-4 and 2-to-3 partitions respectively: the $\sqrt{N}$ scaling of both $g^{H}$ and  $\omega(g^{H})$ persists for these different partitions.

\begin{figure}[!h]
\centering
\includegraphics[height=3.6cm]{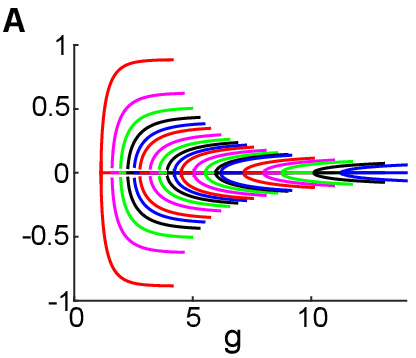} \hspace{0.5cm}
\includegraphics[height=3.6cm]{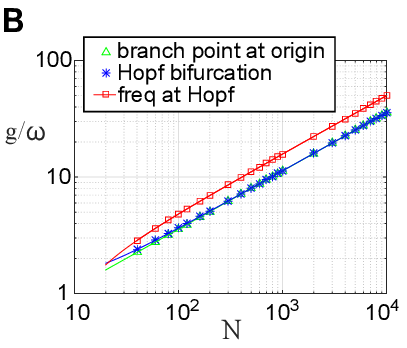}
\caption{Behavior of three-cluster solutions, for equal-size inhibitory clusters ($n_{I_2} = n_{I_1}$). (A) Activity levels on the $I_1/I_2$ solution branch. Colors cycle through $N = 10, 20, 30, 40, 50$, $60, 80, 100, 120, 140$, $160, 200, 240, 280, 300$, $400, 500, 600, 700, 800$, and $1000$ (note: $n_I = N/5$ must be a multiple of 2). (B) Bifurcation values $g^{\ast}$, $g^{H}$, and Hopf frequency $\omega(g^{H})$, as a function of $N$.} \label{fig:bif_n1_2_n2_2}
\end{figure}

\begin{figure}[!h]
\centering
\includegraphics[width=4.2cm]{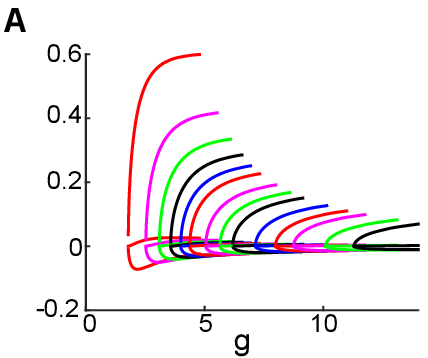} \hspace{0.5cm}
\includegraphics[width=4.2cm]{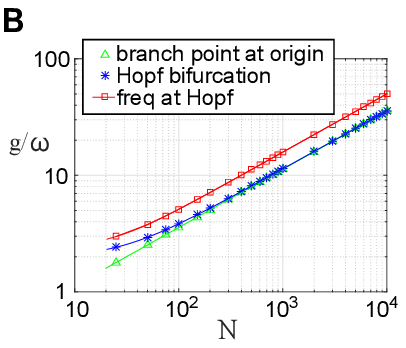}\\
\includegraphics[width=4.2cm]{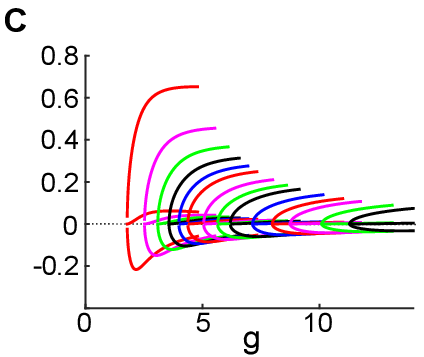} \hspace{0.5cm}
\includegraphics[width=4.2cm]{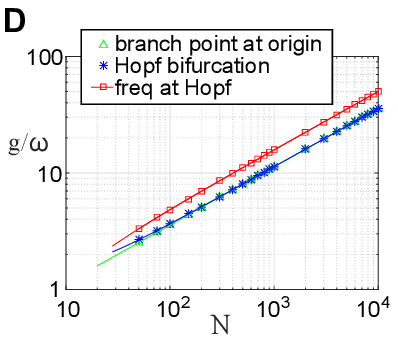}\\
\caption{Behavior of three-cluster solutions as system size $N$ increases. (A-B) clusters where inhibitory cells break into groups with size ratio 1 to 4 ($n_{I_2} = 4 n_{I_1}$).  (A) Activity levels on the $I_1/I_2$ solution branch. Colors cycle through $N = 25, 50, 75, 100$, $125, 150, 200, 250, 300, 400$, $500, 600, 700, 800$, and $1000$ ($n_I = N/5$ must be a multiple of 5).  (B) Bifurcation values $g^{\ast}$, $g^{H}$, and Hopf frequency $\omega(g^{H})$, as a function of $N$. (C-D) Solution branch in which inhibitory cells break into groups with size ratio 2 to 3 ($n_{I_2} = (3/2) n_{I_1}$). (C) Activity levels on the $I_1/I_2$ solution branch. Colors cycle through $N = 25, 50, 75, 100, 125$, $150, 200, 250, 300, 400$, $500, 600, 700, 800$, and $1000$ ($n_I = N/5$ must be a multiple of 5). (D) Bifurcation values $g^{\ast}$, $g^{H}$, and Hopf frequency $\omega(g^{H})$, as a function of $N$.} \label{fig:bif_n1_1_n2_4}
\end{figure}

The $\sqrt{N}$ scaling of $g^{\ast}$, $g^{H}$, and $\omega(g^{H})$ yields insight into the expected behavior of these solutions. First, we should expect these oscillations to become less observable, as $N$ increases; $g^{\ast}$ will eventually reach unphysical values. Second, we should expect the oscillations to become faster as $N$ increases, also eventually reaching an unphysical frequency. Thus, we expect the phenomenon we describe here, to be most relevant for small-to-medium $N$. In the next section, we will show that we can easily find an example for $N=200$; 
 the oscillation period in that example is comparable to the membrane time constant, which is a reasonable upper bound for frequency.

\newpage

\section{Demonstration of relevance to random networks ($\epsilon > 0$)}
We next demonstrate that the bifurcation structure we have described can explain low-dimensional dynamics in example random networks. We 
return to equations \eqref{eqn:Gsys}, \eqref{eqn:G_def} but now let $\epsilon > 0$. The right-hand side of Eqn. \eqref{eqn:Gsys} can be readily shown to be locally Lipschitz continuous in $\mathbb{R}^N$; thus, solutions will vary continuously as a function of parameters (such as $\epsilon$). In particular, we can expect a hyperbolic periodic orbit at $\epsilon = 0$ to persist for some range of $\epsilon \in [0, \epsilon_0)$; here, we will demonstrate this persistence numerically.

We chose parameters $\mu_E = 0.7$, $\sigma_E^2 = 0.625$ and $\sigma_I^2 = 2.5$. (For $\epsilon = 1$, the off-diagonal entries of the resulting random matrices are chosen with the same means and variances as in \cite{RA06}.) We performed a series of simulations in which we fixed $\Avec$, and computed solution trajectories for a range of $\epsilon$ in between $1$ and $0$. As $\epsilon$ decreases, the network connectivity matrix transitions from full heterogeneity (similar to \cite{RA06}), to the deterministic case studied earlier.
\begin{figure}[!h]
\centering
\includegraphics[height=2.7in]{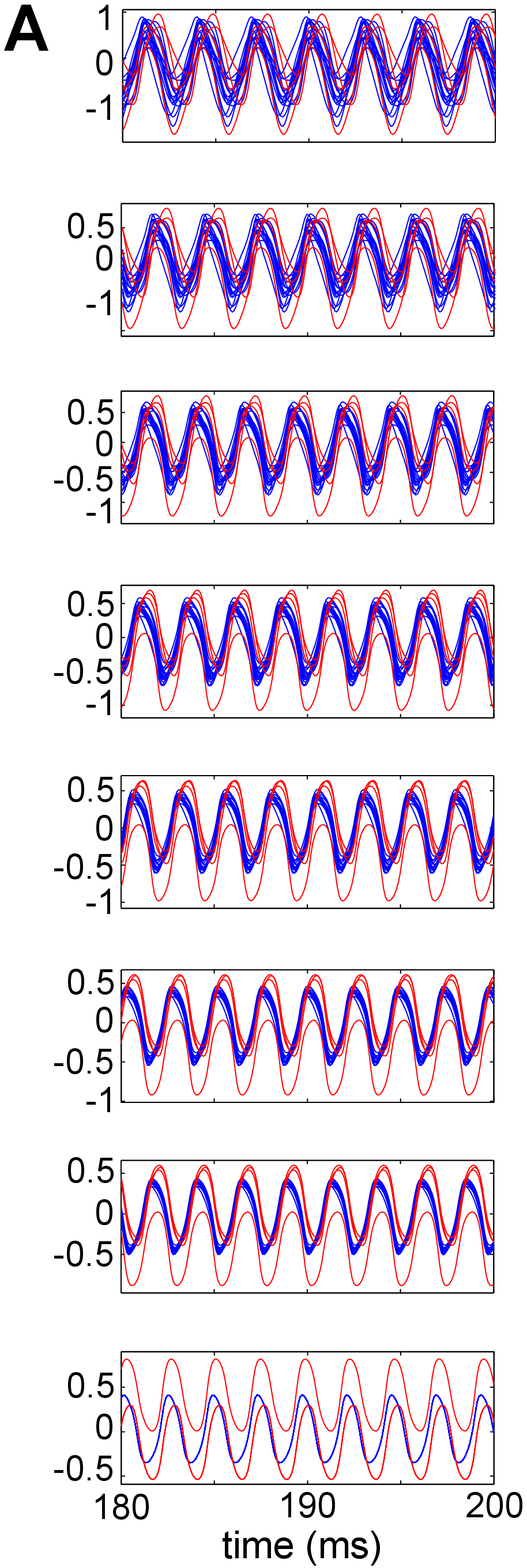}
\includegraphics[height=2.7in]{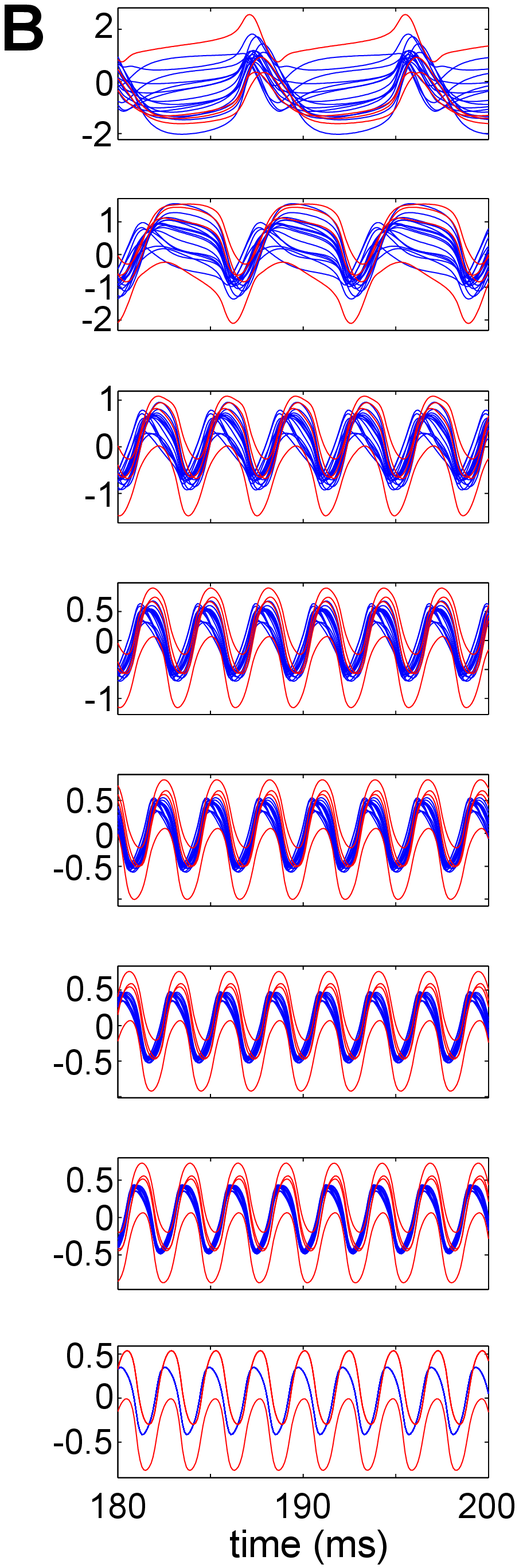}
\includegraphics[height=2.7in]{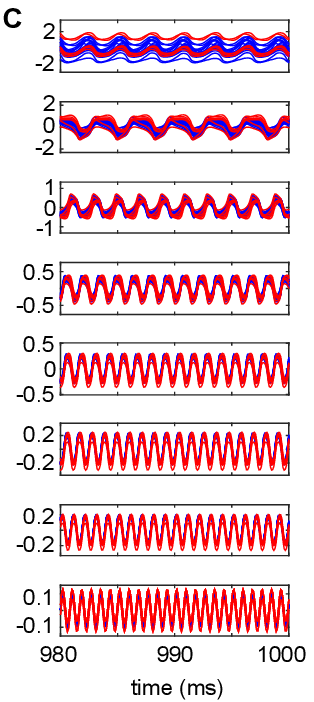}
\includegraphics[height=2.7in]{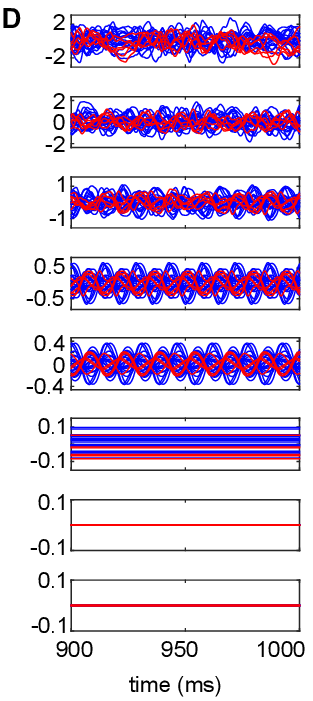}
\includegraphics[height=2.7in]{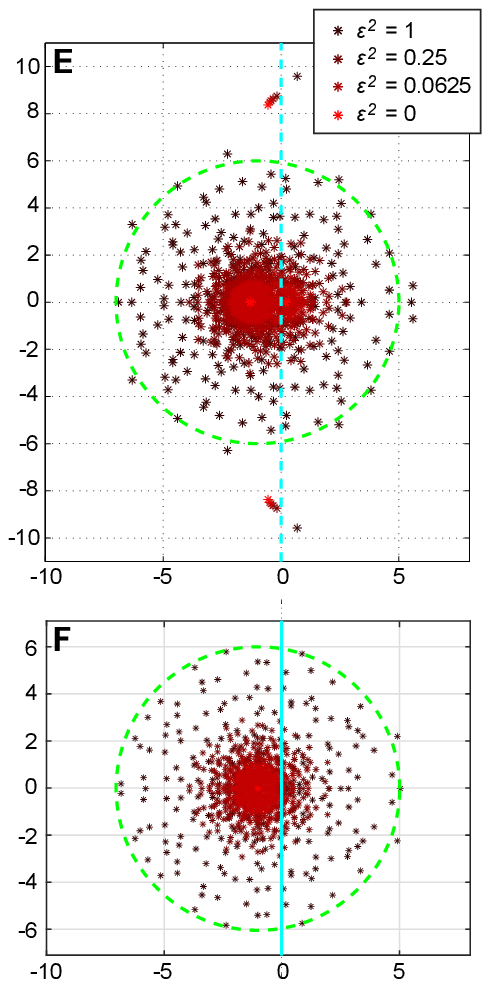}
\caption{(A-B) Solutions for two different networks of size $N = 20$: $g=3$. Here, $\sqrt{N} \Gvec \equiv \Hvec + \epsilon \Avec_{1,2}$. From top to bottom: $\epsilon^2 = 1$, $2^{-1}$, $2^{-2}$, $2^{-3}$, $2^{-4}$, $2^{-5}$, $2^{-6}$, and $0$. (C) Solutions for a network of size $N=200$. The connectivity matrix is given by $\sqrt{N} \Gvec \equiv \Hvec + \epsilon \Avec$, for a single $\Avec$. From top to bottom: $\epsilon^2 = 1$, $2^{-1}$, $2^{-2}$, $2^{-3}$, $2^{-4}$, $2^{-5}$, $2^{-6}$, and $0$) (D) Solutions for a network of size $N=200$, but where $\sqrt{N} \Gvec \equiv \epsilon \Avec$ (i.e. no mean). The random connectivity matrix $\Avec$ is the same as in panel (C).
In (A-B), the traces of $n_E$ excitatory (blue) and $n_I$ inhibitory (red) neurons are shown. In (C-D), only a subset (twenty E and six I cells) is displayed. (E,F) Eigenvalues of the connectivity matrices $\sqrt{N} \Gvec \equiv \Hvec + \epsilon \Avec$ (E) and  $\sqrt{N} \Gvec \equiv \epsilon \Avec$ (F) used in panels (C) and (D). }
\label{fig:N20_traj}
\end{figure}

In Fig. \ref{fig:N20_traj}(A, B) we show two examples of random networks of size $N=20$ and $g=3$. For $\epsilon = 0$ (bottom panel), we indeed see a three cluster solution as predicted. Consistent with our earlier results on limit cycle stability, we see the 3-1 rather than the 2-2 clustering here (in both examples here, $n_{I_1} = 1$ and $n_{I_2} = 3$). The same periodic solution persists as $\epsilon$ increases and is still recognizable at $\epsilon = 1$, illustrated in the top panel (we note that because of the odd symmetry of the governing equations, $-\xvec$ is also a valid trajectory and would appear as a reflection across the time axis). In Fig. \ref{fig:N20_traj}B, the period of the oscillations discernibly increases with $\epsilon$. 

In Fig. \ref{fig:N20_traj}C, we show an example from a larger system, with $N=200$. Here $g=6$; note that a larger coupling value is needed to exceed the bifurcation of the origin at $g^{\ast} = \sqrt{N}/\alpha/\mu_E$. A periodic trajectory is evident in all panels. As in the smaller examples, the period of oscillations increases with $\epsilon$. 

To highlight that this structure is caused by the mean connectivity $\Hvec$, we repeat the sequence of simulations, but integrating the system \textit{without} the mean matrix $\Hvec$. The results are shown in Fig. \ref{fig:N20_traj}D: here the same $\Avec$, initial condition $\xvec_0$, values of $\epsilon$, and coupling parameter $g$ were used; therefore the only difference between each panel in Fig. \ref{fig:N20_traj}D vs. its counterpart in Fig. \ref{fig:N20_traj}C is the absence of the mean connectivity matrix $\Hvec$. Without $\Hvec$, the origin is stable for $\epsilon$ sufficiently small (for $g=6$, $\epsilon^2 < 1/36$); hence the zero solutions in the bottom two panels. As $\epsilon^2$ increases beyond that value we see a fixed point, followed by periodic and then apparently chaotic solutions (for $\epsilon^2 > 2^{-2}$, a decomposition of the trajectories in terms of principal components a large number of orthogonal modes (in excess of 25) ). In addition, the characteristic timescale is much longer than in Fig. \ref{fig:N20_traj}C (note the difference in the time axes). 

Finally, we can contrast the nonlinear behavior with the predicted linear behavior by examining the spectra of the connectivity weight matrices. In Figures \ref{fig:N20_traj}E and  \ref{fig:N20_traj}F, we plot the eigenvalues of $(\Hvec + \epsilon \Avec)/\sqrt{N}$ and $\epsilon \Avec/\sqrt{N}$ respectively, for the specific networks shown in Figure \ref{fig:N20_traj}C-D, and for several values of $\epsilon$. When $\epsilon = 0$, the eigenvalues in Fig. \ref{fig:N20_traj}E cluster into two locations on the real axis, with the exception of one complex pair, as discussed in Example 1 (Fig. \ref{fig:N20_traj}E). In contrast, the eigenvalues in Fig. \ref{fig:N20_traj}F all lie at zero for $\epsilon = 0$. As $\epsilon$ increases, the eigenvalues ``fan out" from their point locations until they fill a disc of radius $g$ (here, $g=6$).  At $\epsilon = 1$, both matrices have dozens of eigenvalues in the right-half plane.


\section{Discussion}
In summary, we studied a family of balanced excitatory-inhibitory firing rate networks that satisfy Dale's Law for arbitrary network size $N$. When there is no variance in synaptic connections --- each excitatory connection has strength $\frac{\mu_E}{\sqrt{N}}$ and each inhibitory connection has strength $\frac{\mu_I}{\sqrt{N}}$ --- we find a family of deterministic solutions whose existence can be inferred from the underlying symmetry structure of the network. These solutions persist in the dynamics of networks with quenched variability --- that is, variance in the connection	 strengths ---   even when the variance is large enough that the envelope of the spectrum of the connectivity matrix approaches that of a Gaussian matrix. This offers a striking example in which linear stability theory is not useful in predicting transitions between dynamical regimes. 
Given the increasing interest in network science, and networked dynamical systems in particular, such observations concerning the impact of symmetry of connectivity can be extremely valuable for studying stability, bifurcations, and reduced-order models.

\subsection*{Role of the deterministic perturbation $\Hvec$}
Gaussian matrices are a familiar object of study in the random matrix community, where Hermitian random matrices are motivated by questions from quantum physics. 
Rajan and Abbott \cite{RA06} studied balanced rank 1 perturbations of Gaussian matrices and found that the spectrum is largely unchanged. These results have since been extended to more general low-rank perturbations \cite{Wei12,tao2013}.  More recently, Ahmadian et al. \cite{Ahmadian_etal_2015} study general deterministic perturbations in the context of neural networks. 
Similarly, recent work has studied extremal values of the spectrum of matrices with modular structure similar to that found here\cite{muir_MF_2015}. Our system is \textit{not} low-rank: in fact, the (seemingly trivial) removal of self-coupling makes the deterministic weight matrix full rank, as we see from Lemma 2.  Using the procedure developed in Ahmadian et al. \cite{Ahmadian_etal_2015}, we can numerically compute the support of spectrum for $\epsilon > 0$ (not shown): as $\epsilon$ grows, this spectral support appears to approach that predicted by a Gaussian matrix or a low-rank perturbation.

However, the more fundamental issue here is that --- except for predicting when the origin becomes unstable --- the spectrum of the full connectivity matrix is not particularly informative about nonlinear dynamics.  Instead, it is the spectrum of the \textit{deterministic} perturbation that emerges as crucial here: the location of the eigenvalues of this matrix can be used to predict the existence of a family of steady states and limit cycles with very specific structure. In the examples presented here (Fig. \ref{fig:N20_traj}), these low-dimensional solutions persist even when $\epsilon$ is large enough that the spectrum of $\Gvec$ is visually indistinguishable from the spectrum of a Gaussian matrix.

It is instructive to 
compare our findings here with the recent results of del Molino et al. \cite{delMolino_etal_PRE_2013}, who consider a balanced excitatory-inhibitory system with a similar $1/\sqrt{N}$ scaling of the mean weights. The authors find a slow, noise-induced oscillation; similar to our results here, this oscillation arises despite an unstable connectivity matrix. Where the two systems differ, is in the deterministic perturbation: del Molino et al. include self-coupling (their deterministic matrix is rank 1), which yields trivial \textit{deterministic} dynamics without a driving current (in the sense of Example 1.2); thus, they do not see the dynamics described here. Conversely, we do not enforce ``perfect balance" $\sum_j \Gvec_{ij} = 0$, which they find is a necessary condition for the slow oscillation to exist; thus we do not see the oscillations described in that paper. Thus, del Molino et al. \cite{delMolino_etal_PRE_2013} and the current work present two distinct examples of dynamics that arise in an excitatory-inhibitory system with $1/\sqrt{N}$ scaling of the mean weights, where linear stability of the connectivity matrix is not informative of the nonlinear dynamics. 


\subsection*{Relationship to other work}
The reduced system described in \S \ref{sec:3cluster} is similar to a simple version of the Wilson-Cowan equations \cite{wc72,wc73}(recently reviewed in \cite{et10,bressloff_2012}).  These equations can be interpreted in terms of coupled neural populations  and can be derived as a mean field limit from large networks. A bifurcation analysis of such a mean-field model was performed recently by Hermann and Touboul \cite{hermann_etal_2012}. Our system differs in two important ways: first, the strong coupling ($1/\sqrt{N}$) means that a factor of $\sqrt{N}$ remains in the reduced equations. Hermann and Touboul, in contrast, pick $J_{ij} \sim N\left(\frac{\bar{J}}{N}, \frac{\sigma}{\sqrt{N}}\right)$; therefore the mean connection strength ($\frac{\bar{J}}{N}$) goes to zero faster than the typical deviation from this mean ($\frac{\sigma}{\sqrt{N}}$): as $N$ becomes large, outgoing synapses are no longer single-signed, in violation of Dale's Law. Similarly, Kadmon and Sompolinsky \cite{kadmon_HS_2015} analyze random diluted networks; they show equivalence to all-to-all Gaussian networks with non-zero mean connections that scale like ($\frac{\bar{J}}{N}$). If the number of synaptic connections per population is held constant, dynamic mean field theory yields predictions for stability which are valid as $N \rightarrow \infty$. 

In contrast, the reduced system in \S \ref{sec:3cluster} does not have a nontrivial limit as $N \rightarrow \infty$, and is not necessarily a limit or a system average; rather, it simply gives reduced dynamics in a specific invariant subspace. Ultimately, every solution of the reduced system is also a perfectly accurate solution of the original system.  The parameter $\beta$ allows a single equation to capture arbitrary bisections of the inhibitory population; in principle, adding more equations would allow further branches to be captured. As another consequence of this scaling, the location of bifurcations $g^{\ast}$ and $g^H$ and the expected frequency of oscillations $\omega(g^H)$, will scale like $\sqrt{N}$; arguably, $g^{\ast}$ and $\omega(g^H)$ will reach unphysical levels, as $N$ becomes large. 

Finally, stronger mean scaling may underlie another difference from previous work; analyzing networks with $1/N$ scaling, other authors have found population-level oscillations via Hopf bifurcations in reduced equations for mean activity \cite{Ginzburg_HS_1994,Brunel_Hakim_1999}. However, in those works the oscillations are not necessarily observable at the level of individual cell activity (particularly strikingly in \cite{Brunel_Hakim_1999}); here, we have distinct cell-level oscillations as well as population-level oscillations. 

Analysis of spontaneous symmetry breaking enjoys a rich history in mathematical biology, and in mathematical neuroscience in particular. 
However, the literature we are aware of identifies symmetry-breaking in a structured network dominated by deterministic behavior. For example, 
symmetry breaking has been hypothesized to underlie the dynamics of visual hallucinations \cite{ErmCow_1979} and ambiguous visual percepts \cite{DieGol_JMN_2014}; central pattern generators which govern rhythmic behaviors of breathing, eating and swimming \cite{But+99,MarBuc_2001, Pearson1993}; and periodic head/limb motions \cite{Gol+PhysD_1998,BuoGol_JMB_2001,GolShiSte_SJAP_2007}. Most recently, Kriener et al. \cite{kriener_etal_2014} investigate a Dale's Law-conforming orientation model, and find that the dynamics are affected by a translation symmetry imposed by the regularity of the cell grid. 
In contrast, the present paper identifies an important role for symmetry in a family of networks usually thought of as dominated by randomness. 
%
%

\subsection*{Future directions}


In this paper, we have focused on analyzing the deterministic system underlying a family of Dale's Law-conforming networks. However, our ultimate interest is in the perturbation away from this system: a full characterization of the dynamics still remains to be completed.  Thus far, we have observed more variable behavior in constrained vs. Gaussian networks: at the same coupling parameter $g$, individual networks display behavior ranging from periodic (as in Fig. \ref{fig:N20_traj}C) to chaotic, suggesting that this task will be more subtle than for Gaussian networks (also see \cite{delMolino_etal_PRE_2013}). Future work will examine this in more detail.

Recent research has focused on the computational power of random networks in the (nominally unpredictable) chaotic regime. Such networks enjoy high computational power because their chaotic dynamics give them access to a rich, complicated phase space, which can be exploited during training to perform complex tasks \cite{SusAbb09,Sussillo_CON_2014}.  
It is an open question as to whether the structure of the networks examined here affects their computational performance on tasks that have been previously studied in Gaussian or other random ensembles. One preliminary study has yielded intriguing results \cite{BarreiroKS_forCNS}: we integrated networks with one of two oscillatory forcing terms $I_1(t)$ and $I_2(t)$, as described in \cite{ostojic2014}, and compared the performance of these networks on two computational tasks: encoding network-averaged firing rate with a sub-population, and discriminating the two inputs in phase space. 
As expected 
Gaussian networks performed worse than constrained networks in encoding population firing rates (similar to what was observed in the balanced networks studied by \cite{ostojic2014}).
However, this difference could not be explained solely by the dimensionality of the solution trajectories (as measured by principal component analysis): constrained networks performed better than Gaussian networks, that required an equal number of principal components to explain their solution trajectories. 
For the second task, 
we observed that for constrained networks, the trajectories associated with $I_1$ and $I_2$ appeared to cluster in distinct regions of principal component phase space; this clustering was not observed for Gaussian networks.

Finally, the ideas explored here can be applied to more general network symmetries: for example, a network with several excitatory clusters and global inhibition, or several weakly connected balanced networks \cite{litwin12}. This will both introduce realism, and allow the exploration of whether there are some universal features that are implied by the broad features of realistic neural network symmetries such as cortex-like excitatory/inhibitory ratios, spatial range specificity of excitatory vs. inhibitory connections, and so forth. We look forward to reporting on this in future work. 

This last direction, in particular, promises to provide further insight into the study of stability and bifurcations in reduced-order models.  The work in this paper
has highlighted how low-dimensional models of high-dimensional networks can 
be used to understand the underlying bifurcation structures resulting from network
connectivity.  Such studies are directly relevant to neuroscience, where input-output 
functionality of extremely high-dimensional networks have been demonstrated to 
be encoded dynamically in low-dimensional subspaces \cite{laurent_NRN_2002,Broome_etal_2006,Yu_etal_JNPhys_2009,machens_etal_JNSci_2010,Churchland_etal_Nature_2012,Shlizerman_etal_SIAP_2012,Shlizerman_etal_FCN_2014}.  We hope that studies such
as this can help highlight both methods for characterizing the 
collective behavior of networked neurons as well as the limits of traditional mathematical
methods in determining stability of such systems.  In either case, the results suggest that 
further study is needed to understand the role of connectivity in driving network
level dynamics.

\bibliographystyle{plain}
\bibliography{POD_bif_ppr_for_arXiv_R1}

\end{document}